\documentclass[%
 aip,
 amsmath,amssymb,
 reprint,%
]{revtex4-1}

\usepackage{graphicx}
\usepackage{dcolumn}
\usepackage{bm}

\usepackage[utf8]{inputenc}
\usepackage[T1]{fontenc}
\usepackage{mathptmx}
\usepackage{etoolbox}
\usepackage[utf8]{inputenc}
\usepackage[english]{babel}
\usepackage[T1]{fontenc}
\usepackage{amsmath}
\usepackage{hyperref}
\usepackage{float}

\usepackage{tikz}
\usepackage{lipsum}

\usepackage{algorithm}
\usepackage{color}
\usepackage{algorithmic}
\usepackage{amsmath}
\usepackage{graphics}
\usepackage{graphicx}
\usepackage{braket}
\usepackage{multirow}
\usepackage{amssymb}
\usepackage{xcolor}
\usepackage{amsthm}
\usepackage{bm}
\usepackage{thmtools}
\declaretheorem{theorem}
\usepackage{array}
\usepackage{url}
\usepackage{hyperref}
\hypersetup{
    colorlinks=true,
    linkcolor=blue,
    filecolor=magenta,      
    urlcolor=magenta,
    citecolor=teal
}
\usepackage{graphicx}
\usepackage{bibentry}
\usepackage{booktabs}
\usepackage{dcolumn}
\usepackage{bm}
\usepackage{blkarray}
\usepackage{xparse}
\usepackage{mathtools}
\usepackage{microtype}
\usepackage{enumerate}

\providecommand{\U}[1]{\protect\rule{.1in}{.1in}}
\newtheorem*{theorem*}{Theorem}

\DeclareMathOperator{\Tr}{Tr}

\newcommand{\outprod}[1]{\ket{#1}\!\!\bra{#1}}

\newcolumntype{P}[1]{>{\centering\arraybackslash}p{#1}}
\makeatletter
\def\@email#1#2{%
 \endgroup
 \patchcmd{\titleblock@produce}
  {\frontmatter@RRAPformat}
  {\frontmatter@RRAPformat{\produce@RRAP{*#1\href{mailto:#2}{#2}}}\frontmatter@RRAPformat}
  {}{}
}%
\makeatother
\begin{document}

\preprint{AIP/123-QED}

\title{Exploring the possibility of a complex-valued non-Gaussianity measure for quantum states of light}
\author{Andrew J. Pizzimenti}
\affiliation{Wyant College of Optical Sciences, The University of Arizona, 1630 E. University Blvd., Tucson, AZ 85721, USA}
\email{ajpizzimenti@arizona.edu}
\thanks{Contributed equally to this work}
\author{Prajit Dhara}
\affiliation{Wyant College of Optical Sciences, The University of Arizona, 1630 E. University Blvd., Tucson, AZ 85721, USA}
\thanks{Contributed equally to this work}
\author{Zacharie Van Herstraeten}
\affiliation{Wyant College of Optical Sciences, The University of Arizona, 1630 E. University Blvd., Tucson, AZ 85721, USA}
\thanks{Contributed equally to this work}
\author{Sijie Cheng}
\affiliation{Department of Electrical and Computer Engineering, The University of Arizona, Tucson AZ, 85721, USA}
\author{Christos N. Gagatsos}
\affiliation{Department of Electrical and Computer Engineering, The University of Arizona, Tucson AZ, 85721, USA}
\thanks{Contributed equally to this work}

\date{\today}

\begin{abstract}
	We consider a quantity that is the differential relative entropy between a generic Wigner function and a Gaussian one. We prove that said quantity is minimized with respect to its Gaussian argument, if both Wigner functions in the argument of the Wigner differential entropy have the same first and second moments, i.e., if the Gaussian argument is the \emph{Gaussian associate} of the other, generic Wigner function. Therefore, we introduce the differential relative entropy between any Wigner function and its Gaussian associate and we examine its potential as a non-Gaussianity measure.  The proposed, phase-space based non-Gaussianity measure is complex-valued, with its imaginary part possessing the physical meaning of the Wigner function's negative volume. At the same time, the real part of this measure provides an extra layer of information, rendering the complex-valued quantity a measure of non-Gaussianity, instead of a quantity pertaining only to the negativity of the Wigner function. We prove that the measure (both the real and imaginary part) is faithful, invariant under Gaussian unitary operations, and find a sufficient condition on its monotonic behavior under Gaussian channels. We provide numerical results supporting aforesaid condition. \color{black} Additionally, we examine the measure's usefulness to non-Gaussian quantum state engineering with partial measurements.
\end{abstract}

\maketitle

\section{Introduction}
\label{sec:intro}
Non-Gaussian states of light have attracted the interest of the broad quantum information community as they possess properties which unlock quantum enhancement in a series of protocols. Particularly, quantum systems that do not exhibit negativity in their Wigner function descriptions, can be simulated efficiently by a classical computer~\cite{Mari2012} and it has been demonstrated that such states do not enable a quantum computational advantage~\cite{Veitch2013}. In the same context, non-Gaussian states are of fundamental interest for quantum computers' architectures, based on all-photonic cluster states~\cite{Pant2019}. Further, non-Gaussian states play a central role in quantum communications~\cite{Niset2009}, while they are useful for all-optical quantum repeaters~\cite{Pant2015,Lo2015,He2020}.

Past, recent, and contemporary works \cite{Andrew2021,Gagatsos2021,Cerf2005,Ourjoumtsev2009,Ra2017,Hiroki2008,Treps2016,Brouri2009,Marek2008,Barnett2018,Treps2018,Arzani2019,Treps2017,Ferrini2017,Su2019,Gagatsos2019} have studied the production of non-Gaussian states utilizing conditional partial measurements. In such schemes a subset of the modes of a multi-mode Gaussian state are projected on non-Gaussian states (for e.g. Fock states, which describe photon counting). This in turn projects the undetected modes into a non-Gaussian state (referred to as the conditional state), unless the conditional state is a result of partial projection on vacuum states. Typically, the task is to optimize the parameters of the Gaussian resource state and to identify the photon counting pattern(s), such that the generated non-Gaussian state is a good approximation, as quantified by some figure of merit, of a non-Gaussian target state.

Therefore, quantifying the non-Gaussian character of a state is essential. This task can be separated into two distinct tasks: To work within the context of a resource theory and establish a non-Gaussianity monotone, or to work within the context of a non-Gaussianity measure. Monotones and measures, are quantities that map a state $\rho$ to a number, i.e., the action of a non-Gaussianity monotone or measure $\mu[.]$ on the state $\rho$ is $\mu[\rho]=c \in \mathbb{C}$. Monotones and measures are thus functionals of the state (ergo our notation using square brackets), i.e., quantities that map the state (written on some basis) or a quasi-probability of the state to a complex number $c$. This number should be informative on the non-Gaussian character of the state.

To explain the difference between monotones and measures, let us take a few steps back. A general resource theory requires the definition of three things: the \emph{resource}, the operations that cannot increase the defined resource (typically called \emph{free operations}), and the states that do not possess the resource (typically referred to as \emph{free states}). In the context of a non-Gaussianity resource theory, one can define the resource to be the negativity of the Wigner function, i.e., Wigner-positive states, even if they are non-Gaussian, are considered free states. The other option is to define the overall non-Gaussianity of the state as a resource, where free states are all Gaussian states. For both approaches, the free operations are all Gaussian protocols: Gaussian unitary operations, partial trace, partial Gaussian measurements, tensor product with a Gaussian state. A proper monotone evaluated for any given state from both of the aforementioned approaches, should then remain invariant under Gaussian unitary operations, should not increase on average under partial trace, be non-increasing on average when part of the state is measured by projections on Gaussian states, and remain invariant when the state under question is composed (by a tensor-product) with a Gaussian state. Further properties of a monotone may include \textit{faithfulness}, i.e., the property that the monotone is equal to zero if and only if the state for which it is evaluated for is a free state. Important works considering the negativity of the Wigner function as a resource include Refs.~\cite{Albarelli2018,Zhuang2018}. The literature on resource theories considering all Gaussian states as free states, is scarce with notable exception Ref.~\cite{PhysRevA.98.022335}. Another successful parallel line of research has been defining experimentally friendly sufficient conditions for a state to be considered quantum non-Gaussian \cite{PhysRevLett.106.200401,PhysRevA.87.062104}. 
Non-Gaussianity \emph{measures} are mathematically more relaxed than their monotone counterparts, while still fully capturing and quantifying the non-Gaussian character of the state under examination.
Non-Gaussianity measures are generally expected to be faithful, invariant under Gaussian unitary operations, and non-increasing under Gaussian channels. The line of research yielded interesting results~\cite{Genoni2007,Genoni2008,Genoni2010,Park2021,Park2021,Chabaud2020Stellar} and in this work, we follow the same path.

The phase-space description of quantum states of light, and particularly the Wigner function, has enabled the detailed study of quantum light. It is well known that the Wigner function constitutes a complete description of a state. Furthermore, the Wigner function can reveal intrinsic properties of a quantum system pertaining to its non-Gaussian character, i.e., exhibiting negative volume, shape of positive and negative regions. Of course, examining the shape of Wigner function assumes that we study single-mode states, which correspond to three-dimensional plots of the corresponding Wigner function. This last remark, motivates further our pursue of a non-Gaussianity measure for generic, multi-dimensional state. Then, the question we pose is if we can use the Wigner function of a state to build a faithful non-Gaussianity measure. It is true that some of the measures already existing in the literature, can be worked out in phase-space. For example, the measure based on the Hilbert-Schmidt distance \cite{Genoni2007,Genoni2010}, essentially requires to calculate the trace of the product of two density operators; something that can be done efficiently in phase-space by integrating together the two corresponding Wigner functions (and taking into account the proper constant in front of said integral, see \cite{Ferraro2005}). Also, recent work has established phase-space based non-Gaussianity measures \cite{Park2021}. In this work, we exploit the recently introduced Wigner entropy \cite{Zach2021,hertzthesis, Cerf2023-gn} and we introduce a complex-valued non-Gaussianity measure whose imaginary part is the total negative volume of the Wigner function corresponding to the state under examination. The real part of the measure provides further information on the non-Gaussian character of the state, and both imaginary and real parts form a faithful non-Gaussianity measure. When the imaginary part of our measure is zero and the real part is non-zero, that is when the Wigner function of the state in question is positive but non-Gaussian, one can further distinguish between states that are just convex mixtures of Gaussian states and states that cannot be expressed as such by taking into account the purity of state.

This paper is organized as follows: In Sec.~\ref{sec:definition} we define the Wigner relative entropy, the Gaussian associate state, and our measure. In Sec.~\ref{sec:faithfulness} we prove that our measure is faithful, i.e., it is equal to zero if and only if the state under question is Gaussian. In Sec.~\ref{sec:gauss_invariance}, we prove that our measure remains invariant under Gaussian unitary operations. In Sec.~\ref{sec:gauss_channels} we give a sufficient condition on the expected behavior of the measure under Gaussian channels. We also provide numerical evidence that our measure is monotonic under Gaussian channels. In Secs. \ref{sec:minimum_value} and \ref{sec:apps}, we examine further properties. Namely, in Sec.~\ref{sec:minimum_value} we employ functional methods to minimize our measure. Section \ref{sec:apps} examines a few non-Gaussian states of interest and the potential of our measure as a computational tool to aid in quantum state engineering based on probabilistic, heralded schemes. Finally, in Sec.~\ref{sec:concl} we outline the main findings of this work and we discuss potential further research paths.

\section{Definition of the measure}
\label{sec:definition}
 In this work, we will only consider states with finite first and second moments such that all displacement vectors and covariance matrices are well defined. Gaussian states are uniquely defined by their displacement vector $\mathbf{d}$ and their covariance matrix $V$ (see for example \cite{Ferraro2005,Weedbrook2012}). Non-Gaussian $N-$mode states have a $2N-$dimensional displacement vector and $2N\times2N-$dimensional covariance matrix, whose definitions are identical to their Gaussian counterparts. That is, for every state $\rho$, the elements $d_i$ and $V_{ij}$ of the displacement vector and covariance matrix respectively, are defined as,
\begin{align}
d_i&=\mathrm{Tr}(\hat{r}_i\rho)\\
V_{ij}&=\frac{1}{2}\mathrm{Tr} \left(\{\hat{r}_i-d_i,\hat{r}_j-d_j\}\rho\right),
\end{align}
where $i,j=1,\ldots,2 N$.

We define the Gaussian associate of $\rho$, denoted as $\rho_G$, the state that has the same displacement vector and covariance matrix as $\rho$. State $\rho$ corresponds to a Wigner function $W(\mathbf{r})$, while state $\rho_G$ corresponds to the Gaussian Wigner function,
\begin{align}
\label{eq:GaussianWigner} W_G(\mathbf{r})=\frac{1}{(2\pi)^N \sqrt{\det \mathbf{V}}}
e^{-\frac{1}{2}(\mathbf{r}-\mathbf{d})^T \mathbf{V}^{-1}(\mathbf{r}-\mathbf{d})}
\end{align}
where $N$ is the number of modes, $\mathbf{r}=(\mathbf{q},\mathbf{p})$ are the coordinates' vector, such that $\mathbf{q}=(q_1,\ldots,q_N)$ is the vector of canonical position and $\mathbf{p}=(p_1,\ldots,p_N)$ is the vector of canonical momentum. We work in the $qqpp$ representation of the state and we set $\hbar=1$.  For the rest of this Section and Section \ref{sec:faithfulness} the barred quantities refer to any Gaussian state while the non-barred quantities refer to generic non-Gaussian states or their Gaussian associates when they also have an index $G$. \color{black}

Typically, a non-Gaussianity measure (henceforth denoted as nGM) for a state $\rho$, is a distance-like measure between the state under examination and its \emph{closest} Gaussian state, intuitively expected to be the state $\rho_G$. For example, the quantum relative entropy and a measure based on the Hilbert-Schmidt distance between $\rho$ and $\rho_G$ have been studied in Ref.~\cite{marian2013relative, Genoni2007,Genoni2008,Genoni2010}.

Using the quantum Wigner entropy, which is \cite{Zach2021,hertzthesis,Cerf2023-gn} defined as,
\begin{align}
\label{eq:WE}h[W]=-\int d^{2N} r W (\mathbf{r}) \ln W(\mathbf{r}),
\end{align}
we define the Wigner relative entropy (WRE) between a Wigner function $W(\mathbf{r})$ and some Gaussian Wigner function of a quantum state $ \overline{W}_G(\mathbf{r}) \color{black}$  as,
\begin{align}
\label{eq:WREtilde}D[W||\overline{W}_G\color{black}] = -h[W]-\int d^{2N} r W (\mathbf{r}) \ln  \overline{W}_G(\mathbf{r}) \color{black}
\end{align}

The Gaussian Wigner function $ \overline{W}_G(\mathbf{r}) \color{black}$ is arbitrary, i.e.\ ,
\begin{align}
    \label{eq:GaussianWignerTilde} 
     \overline{W}_G(\mathbf{r}) \color{black}=\frac{1}{(2\pi)^N \sqrt{\det  \overline{\mathbf{V}}}} \color{black}
    e^{-\frac{1}{2}(\mathbf{r}-\overline{\mathbf{d}} \color{black})^T \overline{\mathbf{V}}^{-1} \color{black}(\mathbf{r}-\overline{\mathbf{d}} \color{black})}
\end{align}
is not assumed to be equal to $W_G(\mathbf{r})$. Instead, we define the quantity,
\begin{align}
    \label{eq:def0}
    \mu[W]=\underset{\overline{W}_G\color{black} \in \mathcal{G}}{\text{min}} D[W||\overline{W}_G\color{black}],
\end{align}
where the functional minimization with respect to $\overline{W}_G\color{black}$ is over the set $\mathcal{G}$ of all Gaussian Wigner functions. Using Eqs.~\eqref{eq:GaussianWigner} and \eqref{eq:GaussianWignerTilde}, Eq.~\eqref{eq:WREtilde} gives,
\begin{align}
    \nonumber D[W||\overline{W}_G\color{black}] = -h[W]+N \ln (2\pi)+\frac{1}{2} \ln \det \overline{\mathbf{V}} \color{black}\\
    \label{eq:WRE2} + \frac{1}{2}\text{Tr}(\mathbf{V}\overline{\mathbf{V}}^{-1} \color{black})+\frac{1}{2}(\mathbf{d}-\overline{\mathbf{d}} \color{black})^T \overline{\mathbf{V}}^{-1} \color{black} (\mathbf{d}-\overline{\mathbf{d}} \color{black}).
\end{align}
In Eq. \eqref{eq:WRE2}, the matrix $\overline{\mathbf{V}}^{-1} \color{black}$ is positive definite (as it is the inverse of a quantum state's covariance matrix). Therefore, we have the inequality,
\begin{align}
(\mathbf{d}-\overline{\mathbf{d}} \color{black})^T \overline{\mathbf{V}}^{-1} \color{black} (\mathbf{d}-\overline{\mathbf{d}} \color{black})\geq 0
\end{align}
with equality to zero if and only if $\mathbf{d}=\overline{\mathbf{d}} \color{black}$, which means that $ \overline{W}_G(\mathbf{r}) \color{black}$ and $W(\mathbf{r})$ must have identical displacement vectors. Therefore, from Eq. \eqref{eq:WRE2}, we now have to minimize the quantity,
\begin{align}
\nonumber D[W||\overline{W}_G\color{black}] =& -h[W]+N \ln (2\pi)+\frac{1}{2} \ln \det \overline{\mathbf{V}} \color{black}\\
\label{eq:WRE3} &+ \frac{1}{2}\text{Tr}(\mathbf{V}\overline{\mathbf{V}}^{-1} \color{black}).
\end{align}
We proceed to show that Eq.~\eqref{eq:WRE3} is minimized if and only if $\overline{\mathbf{V}} \color{black}=\mathbf{V}$. Using the right-hand side of Eq.~\eqref{eq:WRE3} it suffices to prove the following inequality,
\begin{align}
\nonumber &\ln \det \overline{\mathbf{V}} \color{black}+\text{Tr}(\mathbf{V}\overline{\mathbf{V}}^{-1} \color{black}) \geq \ln \det \mathbf{V}+\text{Tr}(\mathbf{VV}^{-1}) 
\\\Leftrightarrow\quad& 
\nonumber \text{Tr}(\mathbf{V}\overline{\mathbf{V}}^{-1} \color{black}) \geq \ln \det (\mathbf{V}\overline{\mathbf{V}}^{-1} \color{black})+2N 
\\
\nonumber
\Leftrightarrow\quad & \text{Tr}(\mathbf{V}^{\frac{1}{2}} \overline{\mathbf{V}}^{-1}\color{black} \mathbf{V}^{\frac{1}{2}}) \geq \ln \det (\mathbf{V}^{\frac{1}{2}} \overline{\mathbf{V}}^{-1} \color{black} \mathbf{V}^{\frac{1}{2}})+2N
\\
\Leftrightarrow\quad &\text{Tr}\mathbf{A} \geq \ln \det \mathbf{A}+2N.
 \label{eq:Ineq}
\end{align}
When going from the second to the third line of \eqref{eq:Ineq}, in addition to to the cyclic property of the trace and determinant, we use the fact that $\mathbf{V} = \mathbf{V}^{\frac{1}{2}} \mathbf{V}^{\frac{1}{2}}$ for a positive definite matrix $\mathbf{V}$, $\mathbf{V}^{\frac{1}{2}} $ its self being positive definite as well. The positive definite matrix $A = \mathbf{V}^{\frac{1}{2}} \overline{\mathbf{V}}^{-1}\color{black} \mathbf{V}^{\frac{1}{2}}$ has eigenvalues $\{\lambda_1,\ldots,\lambda_{2N}\}$ with $\lambda_i>0,\, \forall\, i=1,\ldots,2N$. For any positive $x$, it is always true that $x-1 \geq \ln x$. Therefore, for the eigenvalues of $A$, we have,
\begin{align}
\nonumber 
&\sum_{i=1}^{2N}\left(\lambda_i-1\right) \geq \sum_{i=1}^{2N}\ln \lambda_i 
\\\Leftrightarrow \nonumber\quad&\sum_{i=1}^{2N}\lambda_i - 2N \geq \ln\prod_{i=1}^{2N}\lambda_i 
\\\Leftrightarrow\quad&\text{Tr}\mathbf{A} \geq \ln \det \mathbf{A} + 2N,
\end{align}
which is the desired inequality \eqref{eq:Ineq}. Therefore, $D[W||\overline{W}_G\color{black}]$ is minimized if $ \overline{W}_G(\mathbf{r}) \color{black}$ is the Gaussian associate of $W(\mathbf{r})$, i.e., $ \overline{W}_G(\mathbf{r}) \color{black}=W_G(\mathbf{r})$.

Hence, with the proof above Eq.~\eqref{eq:def0} gives,
\begin{align}
\label{eq:def1} \mu[W]= D[W||W_G].
\end{align}
We note that there is only one argument in $\mu[.]$ as $W_G(\mathbf{r})$ has a functional dependence on $W(\mathbf{r})$ and is uniquely defined. Using Eq.~\eqref{eq:WREtilde}, we write,
\begin{subequations}
    \begin{align}
    \label{eq:muN}	\mu[W] &= -h[W]-\int d^{2N}r W(\mathbf{r}) \ln W_G(\mathbf{r}) \\
    \label{eq:muN2}&= -h[W]-\int d^{2N}r W_G(\mathbf{r}) \ln W_G(\mathbf{r})\\
    \label{eq:muN3}&= -h[W]+\ln \left((2\pi e)^N \sqrt{\det \mathbf{V}}\right),
    \end{align}
\end{subequations}

where we have used the fact that,
\begin{align}
\nonumber \ln W_G(\mathbf{r})&=-\ln \left((2\pi)^N \sqrt{\det \mathbf{V}}\right)\\
&-\frac{1}{2}(\mathbf{r}-\mathbf{d})^T \mathbf{V}^{-1}(\mathbf{r}-\mathbf{d}),
\end{align}
rendering the integrals in Eqs.~\eqref{eq:muN} and~\eqref{eq:muN2} equal, i.e., only the first and second moments will play role in said integrals. Equation~\eqref{eq:muN2} is equal to Eq. \eqref{eq:muN3} as straightforward evaluation gives,
\begin{align}
\nonumber \frac{1}{2} \int d^{2N}r W_G(\mathbf{r}) (\mathbf{r}-\mathbf{d})^T \mathbf{V}^{-1}(\mathbf{r}-\mathbf{d})=-N.
\end{align}
It is of central importance to discuss how to deal with Wigner functions that posses negative volume. We can write any Wigner function as,
\begin{align}
W(\mathbf{r})=|W(\mathbf{r})|e^{i (2k+1)\pi \phi(\mathbf{q},\mathbf{p})},
\end{align}
where $k\in \mathbb{Z}$ and,
\begin{align}
\label{eq:phi}	\phi(\mathbf{r}) = \left\{ \begin{array}{ll}
1,& W(\mathbf{r})<0 \\
0,& W(\mathbf{r}) \geq 0.
\end{array}
\right.
\end{align} is the marker function which distinguishes between the positive and negative functional values in the domain. Therefore, $\ln W(\mathbf{r})=\ln |W(\mathbf{r})|+i(2k+1)\pi \phi(\mathbf{r})$ where $k$ defines the branch of the complex logarithm. We choose to work with $k=-1$ and we rewrite Eq.~\eqref{eq:muN} as,
\begin{align}
\label{eq:muRI}	\mu[W]= \text{Re}\mu[W]+i \text{Im}\mu[W],
\end{align}
where,
\begin{align}
\label{eq:Remu}	\text{Re}\mu[W] &= -\text{Re}h[W]+h[W_G]\\
\nonumber \text{Im}\mu[W] &= -\pi \int d^{2N} r W(\mathbf{r}) \phi(\mathbf{r}) \\ 
\label{eq:Immu}  &=  \pi |\mathcal{V}_-|
\end{align}
and
\begin{align}
\label{eq:Remu3} \text{Re}h[W] = -\int d^{2N} r W(\mathbf{r}) \ln |W(\mathbf{r})|,
\end{align}
\begin{subequations}
  \begin{align}
h[W_G] &=- \int d^{2N} r W(\mathbf{r}) \ln W_G(\mathbf{r})\\
&=-\int d^{2N} r W_G(\mathbf{r}) \ln W_G(\mathbf{r})\\
&=\ln \left((2\pi e)^N \sqrt{\det \mathbf{V}}\right),
\end{align}  
\end{subequations}
where $|\mathcal{\mathcal{V}}_-|$ is the total negative volume of the Wigner function $W(\mathbf{r})$.

We note the following set of properties for the nGM $\mu[W]$,
\begin{theorem}
    $\mu[W]$ evaluates to zero if and only if $W$ is a Gaussian Wigner function.
\end{theorem}
\begin{proof}
    See Section~\ref{sec:faithfulness} for proof.
\end{proof}

\begin{theorem}
    $\mu[W]$ is invariant under Gaussian unitary operations on the Wigner function $W$ and its underlying state.
\end{theorem}
\begin{proof}
    See Section~\ref{sec:gauss_invariance} for proof and  associated discussion.
\end{proof}

\begin{theorem}
    If Eq. \eqref{eq:condition_gaussian_channel} is satisfied, $\mathrm{Re}\mu[W]$ decreases under the action of a Gaussian channel on the Wigner function $W$ and its underlying state.
\end{theorem}
\begin{proof}
    See Section~\ref{sec:gauss_channels} for proof and  associated discussion.
\end{proof}

\begin{theorem}
    The Gaussian associate of any Wigner function is the only physical critical point and it corresponds to a minimum of $\mathrm{Re}\mu[W]$.
\end{theorem}
\begin{proof}
    See Section~\ref{sec:minimum_value} for proof and  associated discussion.
\end{proof}

\section{Faithfulness}
\label{sec:faithfulness}
By \emph{faithfulness} we mean that $W(\mathbf{r})$ is Gaussian, if and only if $\mu[W]=0$.
We clarify that the meaning of $\mu[W]=0$ is that both its real and imaginary parts are simultaneously equal to zero.

\emph{Proof of direct statement:} Let $W(\mathbf{r})= \overline{W}_G(\mathbf{r}) \color{black}$, i.e., the Wigner function under consideration is Gaussian, but not necessarily equal to its Gaussian associate $W_G(\mathbf{r})$. The Gaussian associate of $W_G(\mathbf{r})$ is equal to itself. Since a Gaussian Wigner function is positive over its entire domain, i.e., $|W_G(\mathbf{r})|=W_G(\mathbf{r})$, we get,
\begin{eqnarray}
\text{Im}\mu[ \overline{W}_G\color{black}]=\pi |\mathcal{V}_{-}|=0.
\end{eqnarray}
From said fact and Eqs.~\eqref{eq:WE} and~\eqref{eq:Remu3}, we get $\text{Re}h[ \overline{W}_G \color{black}]=h[ \overline{W}_G \color{black}]$. Hence using Eq.~\eqref{eq:Remu} we have,
\begin{align}
\label{eq:hG} 
\text{Re}\mu[ \overline{W}_G \color{black}] &= -\text{Re}h[ \overline{W}_G \color{black}]+h[ \overline{W}_G \color{black}]=0.
\end{align}

\emph{Proof of converse statement:} For some Wigner function $W(\mathbf{r})$, let $\mu[W]=0$ , i.e.,
\begin{subequations}
    \begin{align}
    \label{eq:remu=0} \text{Re}\mu[W]&=0\\
    \label{eq:immu=0} \text{Im}\mu[W]&=0.
    \end{align}
\end{subequations}
Equation~\eqref{eq:immu=0} imposes that $W(\mathbf{r})$ is positive over its entire domain, 
which implies that $W(\mathbf{r})$ is well-defined probability distribution function i.e.\ $W(\mathbf{r}) \equiv W_{>0}(\mathbf{r})$ . Consequently, $W(\mathbf{r})$ can be analyzed using classical probability theory, even if the state corresponding to $W(\mathbf{r})$, is a quantum state with no classical description. Therefore, $\text{Re}\,\mu[W_{>0}]$ represents a classical relative entropy quantity being equal to zero, 
\begin{align}
D[W_{>0}|| \overline{W}_G \color{black}]=0
\end{align}
which necessarily implies that $W_{>0}(\mathbf{r})= \overline{W}_G(\mathbf{r}) \color{black}$. It is worthwhile to note that given the natural emergence of $\mathrm{Im}\,\mu[W]$ as the total negative volume of the $W(\mathbf{r})$ ensures the faithfulness property of the complex-valued $\mu[W]$.

\section{Invariance under Gaussian unitary operators}
\label{sec:gauss_invariance}
 $\text{Re}h[W]$ (and consequently $h[W_G]$) are proven to be invariant under symplectic transformations in Ref.~\cite{hertzthesis}. Equivalently, Gaussian unitary operations acting on the state or to the corresponding Wigner function $W(\mathbf{r})$ leave  $\text{Re}h[W]$ unchanged.  Symplectic transformations also conserve the negative volume of a Wigner function, i.e., $\text{Im}\mu[W]$ remains invariant when the corresponding state undergoes a Gaussian unitary transformation. Therefore $\mu[W]$ is invariant under Gaussian unitary operations.

 \section{Behavior under Gaussian channels}
\label{sec:gauss_channels}
\subsection{Sufficient Condition}
 \begin{figure}[!]
    \centering
    \includegraphics[width=\columnwidth]{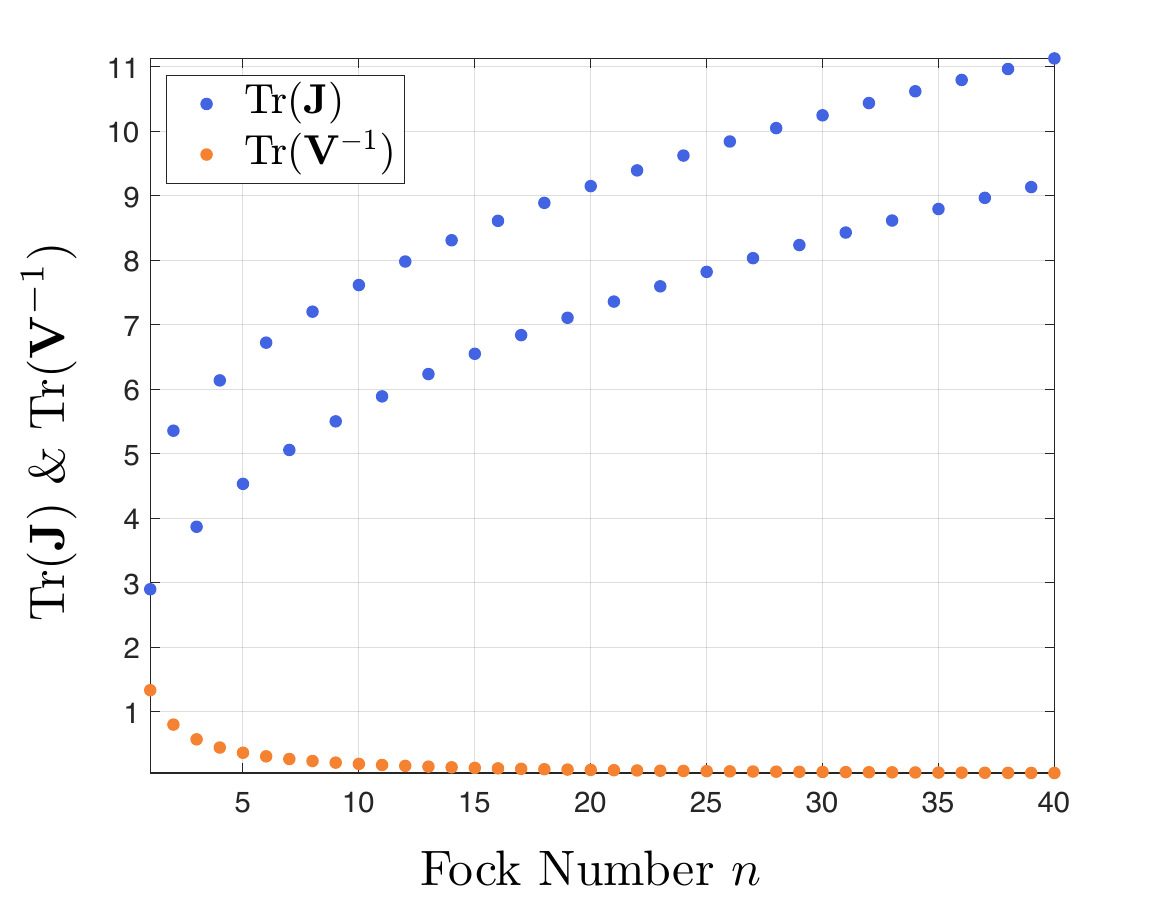}
    \caption{Plot of the trace of the Fisher information matrix $\text{Tr}(\mathbf{J})$ (blue) and $\mathrm{Tr} (\mathbf{V}^{-1})$ (orange) for the Fock states from $n=0$ up to $n=40$.}
    \label{fig:fisher_info_fock_states}
\end{figure}
\label{sec:GaussianChannelAnalytics}

 Intuitively, we expect the value of a non-Gaussianity measure to decrease under the action of a Gaussian channel. Namely, we expect such channels to yield a \emph{more Gaussian} output compared to its input, since it has been mixed with a Gaussian environment through some Gaussian unitary.
In our concern, we study a complex-valued non-Gaussianity measure, and should then check separately the behaviour of the real and imaginary parts.
Fortunately for us, the imaginary part of $\mu[W]$ is proportional to the negative volume of $W$, which has previously been studied and shown to be non-increasing under the action of a Gaussian channel \cite{Albarelli2018,Zhuang2018}.
This section will thus be devoted to the evolution of $\mathrm{Re}\mu[W]$ when $W$ evolves in a Gaussian channel.

 In Appendix \ref{apd:evolution_measure_gaussian_channels}, we present an equivalent condition for the real part of the measure to decrease under any Gaussian channel.
The condition reads as follows,
 \begin{align}
    \mathbf{V}\geq \mathbf{J}^{-1} \Leftrightarrow \mathbf{V}^{-1}\leq \mathbf{J}
    \label{eq:condition_gaussian_channel}
\end{align}
 where $\mathbf{V}$ is the covariance matrix of $W(\mathbf{r})$ and $\mathbf{J}$ is the Fisher information matrix corresponding to $W(\mathbf{r})$. 
  The Fisher information matrix, typically defined with respect to a parameter of interest in metrology, in this case is with respect to the phase space variables $\mathbf{r}$ (See Eq. \eqref{eq:FIwrtW}). \color{black} 
The equivalent condition from which Eq. \eqref{eq:condition_gaussian_channel} stems from is (see Appendix \ref{apd:evolution_measure_gaussian_channels}),
\begin{eqnarray}
\label{eq:SuffAndNec}\text{Tr}[\mathbf{U}(\mathbf{V}^{-1}-\mathbf{J})] \leq 0,
\end{eqnarray}
for all quantum covariance matrices $\mathbf{U}$. For single-mode channels, if $\mathbf{U}\propto\mathbf{I}$ (with a positive constant), the channel can be a Gaussian phase-invariant channel, i.e., the pure loss channel or the quantum-limited amplifier. For said case, Eq. \eqref{eq:SuffAndNec} becomes,
\begin{eqnarray}
    \text{Tr}(\mathbf{V}^{-1}) \leq \text{Tr}(\mathbf{J}),
\end{eqnarray}
which is what we plot in Fig. \ref{fig:fisher_info_fock_states}, for the special case of Fock states $|n\rangle$, known to possess negative regions in their Wigner representations for $n\geq 1$. The approach to compute $\mathbf{J}$ for Fock states is given in Appendix \ref{app:NumericalMethods} and is based on computing the principal value of the involved integrals when is necessary.

Interestingly, Eq. \eqref{eq:condition_gaussian_channel} corresponds precisely to the Cram\'er-Rao bound \cite{Cramer1999-bs, Rao1992-jo}, and is thus always satisfied for non-negative Wigner functions, as they correspond to genuine probability distributions. For Wigner functions that take negative values, we were not able conclude that Eq.~\eqref{eq:condition_gaussian_channel} always holds, nor we have a reason to believe that it does for all states. Nevertheless, our numerical study is partially presented in Fig.~\ref{fig:fisher_info_fock_states} illustrating that Eq.~\eqref{eq:SuffAndNec} is satisfied for Fock states under the action of certain Gaussian channels, i.e., when $\mathbf{U}\propto \mathbf{I}$. Furthermore, in Section~\ref{sec:GaussianChannelNumerics} we provide more numerical evidence on the decreasing behavior of $\mathrm{Re}\mu[W]$ under Gaussian channels.

\subsection{Numerical Evidence}
\label{sec:GaussianChannelNumerics}

In this section we show that $\mu[W]$, for a set of random $d$-dimensional states, is non-increasing under a thermal loss channel. The action of a thermal loss channel is to combine the input state, $\rho_0$, and a thermal state with mean photon number $\overline{n}_B$, represented by $\rho_{\text{th}}(\overline{n}_B)$, on a beam-splitter of transmissivity $\tau$. Tracing out the environmental mode leaves us with the channel output, $\rho_f$. Letting $\overline{n}_B$ go to zero reduces the thermal loss channel to the pure loss channel and letting $\tau$ go to $1$ corresponds to applying identity to $\rho_0$ . We choose to work with the Kraus operators representation instead of checking the vallidity of Eq. \eqref{eq:condition_gaussian_channel}. The thermal loss channel may be decomposed into a pure loss channel of transmissivity $\eta = \tau/G$ and quantum-limited amplifier with gain $G = 1 + (1-\tau)\overline{n}_B$ \cite{ivan2011operator}. The thermal loss channel output is then given by  $\rho_f = \sum_{k,l = 0}^{\infty}\hat{B}_k \hat{A}_l \rho_0 \hat{A}_l^{\dagger}\hat{B}_k^{\dagger}$ where the Kraus operators of the pure loss channel, $\hat{A}_l$, and quantum-limited amplifier, $\hat{B}_k$, are defined as 
\begin{subequations}
\begin{align}
    \hat{A}_l &= \sqrt{\frac{(1-\eta)^l}{l!}} \eta^{\hat{n}/2} \hat{a}^{l}\\
    \hat{B}_k &= \sqrt{\frac{1}{k!}\frac{1}{G}\Big(\frac{G-1}{G}\Big)^k} (\hat{a}^{\dagger})^{k}G^{-\hat{n}/2}.
\end{align}
\label{eq:KrausOperators}
\end{subequations}

We prepare $100$ random $d$-dimensional states  for $d \in \{2, 3, 4\}$, i.e. $300$ states in total, by generating random diagonal states in generic basis $\{\ket{\overline{0}}\!\!\bra{\overline{0}}, \ket{\overline{1}}\!\!\bra{\overline{1}}, \ldots, \ket{\overline{d-1}}\!\!\bra{\overline{d-1}}\}$ and then applying a Haar-random unitary to them \cite{zyczkowski1994random, zyczkowski1998volume}. By setting the basis states to different Fock states we extend the total number of unique states tested to $900$. In Fig.~\ref{WREQuditsThermal}, we plot $\mu[W]$ for a subset of the $900$ random states transmitted through a thermal loss channel of $\overline{n}_B = 0.001$, as a function of channel transmissivity $\tau$.
\begin{figure*}[!ht]
    \centering
    \includegraphics[width=\textwidth]{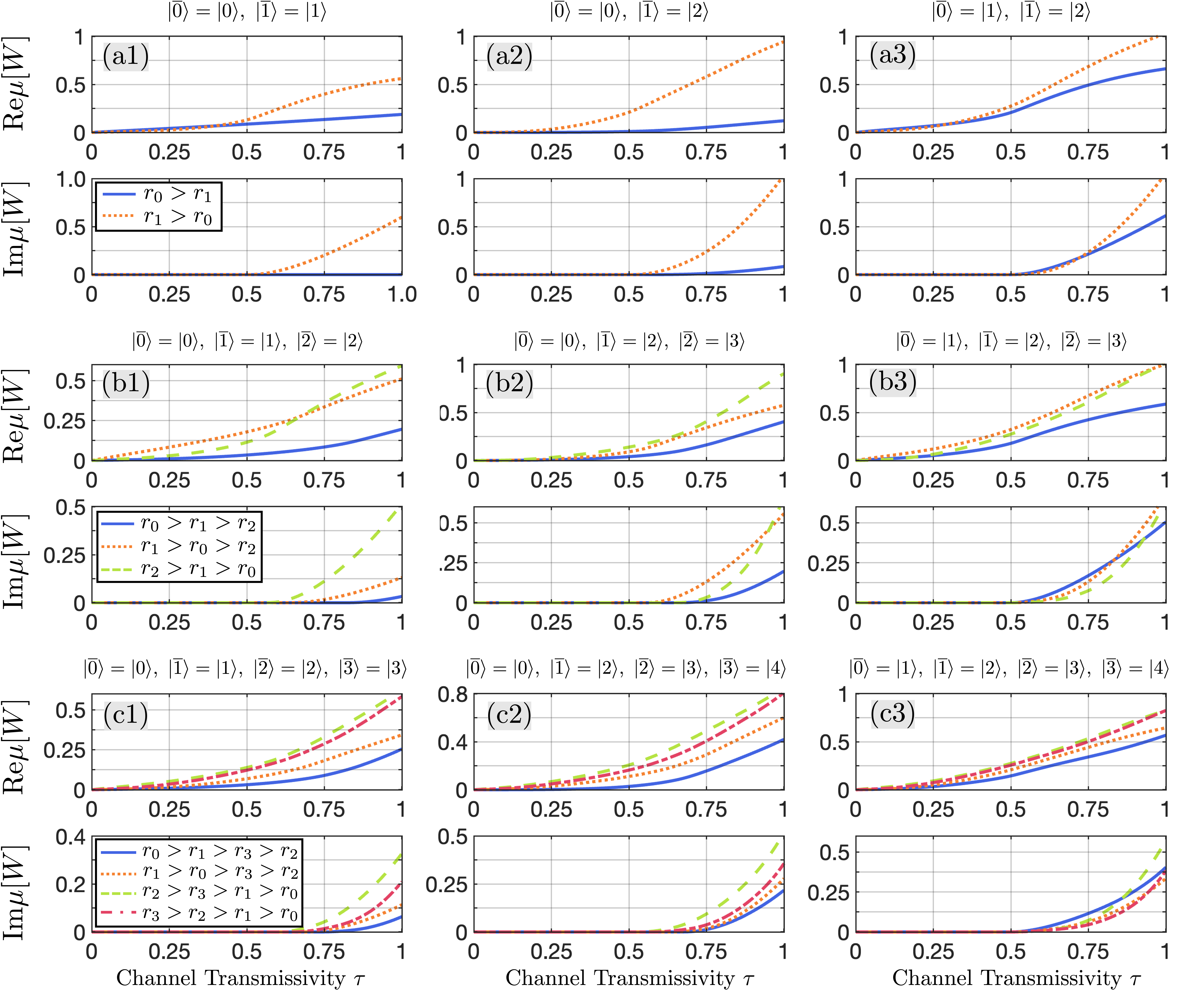}
    \caption{Non-Gaussianity measure evaluated for a subset of the $900$ total Haar-random $d$-dimensional states (a) $d= 2$, (b) $d=3$, and (c) $d = 4$  after traversing a thermal loss channel with temperature $\overline{n}_B = 0.001$ and transmissivity $\tau$.  The Haar random states are created by applying a a Haar-random unitary to an initial state $\hat{\rho} = r_0\ket{\overline{0}}\!\!\bra{\overline{0}} + r_1 \ket{\overline{1}}\!\!\bra{\overline{1}} + \ldots + r_{d-1}\ket{\overline{d-1}}\!\!\bra{\overline{d-1}}$. The coefficients $\{r_0, r_1, \ldots , r_{d-1}\}$ satisfy the inequality relations show in the insets. When evaluating the non-Gaussianity measure, the logical basis states are taken to be the Fock states (a1,b1,c1) $\{\ket{0}, \ldots, \ket{d-1}\}$, (a2,b2,c2) $\{\ket{0}, \ket{2}, \ldots, \ket{d}\}$, and (a3,b3,c3) $\{\ket{1}, \ldots, \ket{d}\}$.}
    \label{WREQuditsThermal}
\end{figure*}

For the states shown in Fig.~\ref{WREQuditsThermal} we see the action of the thermal loss channel only serves to reduce both $\text{Re}\mu[W]$ and $\text{Im}\mu[W]$ of the input states. This trend holds true for all $900$ Haar-random states supporting that our measure is non-increasing under the thermal loss channel, for this set of states.  

As further numerical evidence that $\mathrm{Re}\mu[W]$ is non-increasing under thermal loss, we study the evolution of Schr\"{o}dinger cat and GKP states (see Sec.~\ref{sec:apps}) under the thermal loss channel. For these more complicated non-Gaussian states we utilize phase space methods to simulate their evolution under the thermal loss channel. Specifically, if we define the action of re-scaling operator $\mathcal{L}_{s}$ on a Wigner function as 
\begin{align}
\label{eq:ScalingOperator}
    \mathcal{L}_{s}[W](\mathbf{r}) = \frac{1}{s^{2}}W\Big(\frac{\mathbf{r}}{s}\Big),
\end{align}
then the Wigner function of any state which is subjected to a thermal loss channel of transmissivity $\tau$ and mean photon number $\overline{n}_B$ can be expressed as the convolution of the re-scaled original state Wigner function with the re-scaled Wigner function of a thermal state. That is,
\begin{align}
\label{Eq.PhaseSpaceThermalLoss}
    W_{\text{out}}(\mathbf{r}) = \mathcal{L}_{\sqrt{\tau}}[W_{\text{in}}](\mathbf{r}) \ast \mathcal{L}_{\sqrt{1-\tau}}[W_{\overline{n}_B}](\mathbf{r}).
\end{align}

In Fig.~\ref{WRECatandGKPThermal} we plot $\mu[W_{\text{out}}]$ as a function of $\tau$ for input odd (A1) and even (A2) cat states of varying coherent amplitude $\alpha \in \mathbb{R}$ and input logical zero (B1) and one (B2) GKP states of varying squeezing $\Delta$.
\begin{figure*}[!]
    \centering
    \includegraphics[width=\textwidth]{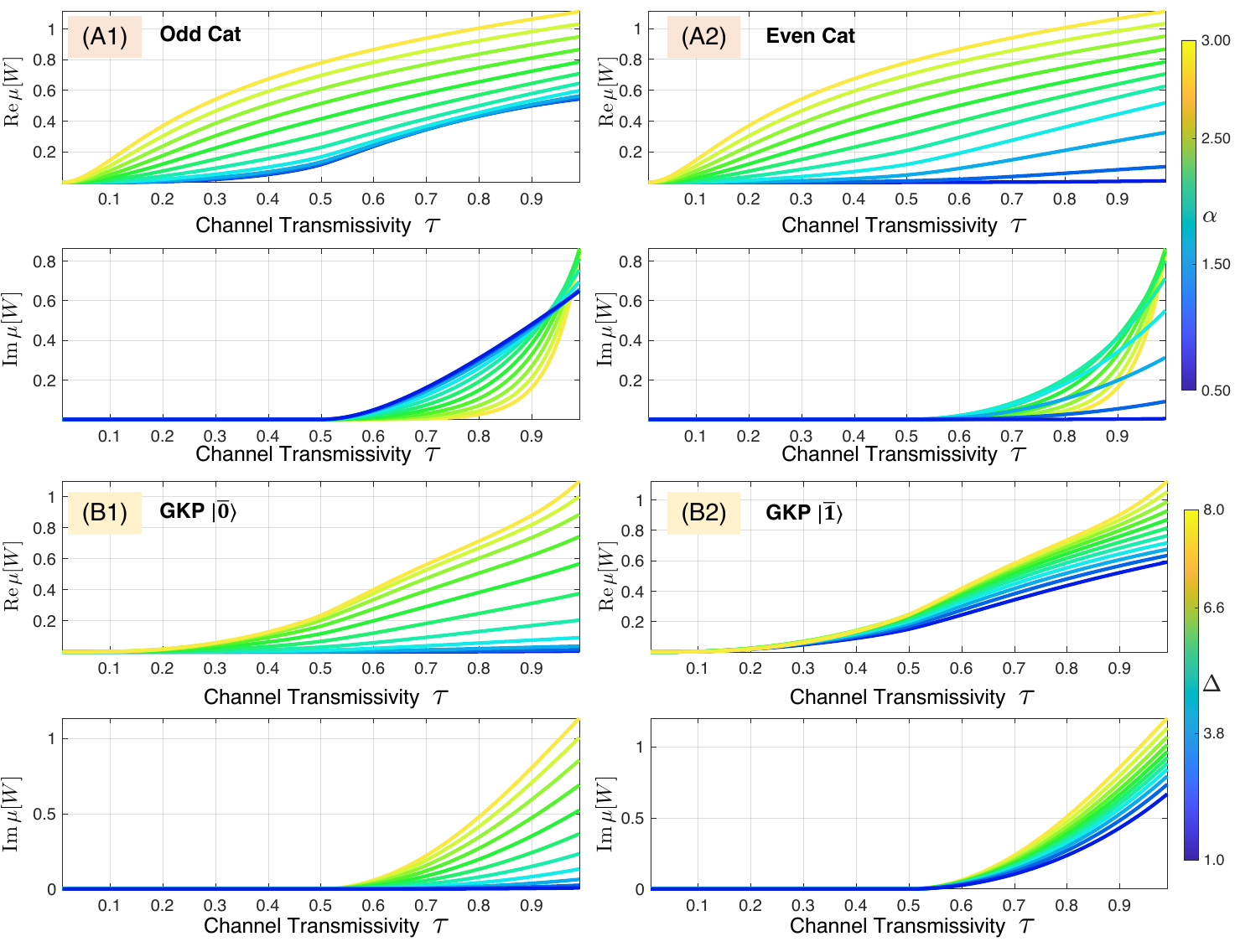}
    \caption{Non-Gaussianity measure of  (A1) odd cat, (A2) even cat, logical zero GKP (c), and logical one GKP (d) states after traversing a thermal loss channel of transmissivity $\tau$ and mean photon number $\overline{n}_B = 0.001$.}
    \label{WRECatandGKPThermal}
\end{figure*}
 

As can be seen in Fig. \ref{WRECatandGKPThermal}, $\mu[W]$ only decreases with $\tau$ for the states in question. A more in-depth examination of how the coherent amplitude of cat states and squeezing parameter of GKP states effects the measure's evolution under Gaussian channels is the subject of future work.

\section{Minimum value of the measure}
\label{sec:minimum_value}
In this paper we define a nGM which is complex and therefore, proving that its real and imaginary parts are strictly positive, is unnecessary; Indeed, one can redefine the nGM as $|\mu[W]|^2$ which is by definition always non-negative, and all properties and numerical simulations in Secs. \ref{sec:faithfulness}, \ref{sec:gauss_invariance}, \ref{sec:gauss_channels}, and \ref{sec:apps} hold without any qualitative difference. However, since we explore a new quantity, i.e., essentially a relative entropy between Wigner functions, we will undertake the task to examine the minimum value of the real and imaginary parts of $\mu[W]$.
\subsection{Real part}\label{subsec:MinReal}
We employ functional methods to find (local) minima of the functional $\mathrm{Re}\mu[W]$, and show that they coincide with Gaussian states. We consider the functional,
\begin{align}
\label{eq:FuncRe1}
    \nonumber R[W] =\text{Re}\mu[W]+\lambda_1 G^{(1)}[W] \\
    +\sum_{i=1}^{N}\mu_iG_i^{(2)}[W]+\sum_{i,j=1}^{ 2 N }   \Lambda_{ij}G_{ij}^{(3)}[W],
\end{align}
where $\text{Re}\mu[W]$ is given in Eq.~\eqref{eq:Remu}, $\{\lambda_1, \mu_i, \Lambda_{ij}\},\ i,j=1,\ldots, 2 N $ are the Lagrange multipliers to impose constraints in a functional form. The first constraint requires that the $W(\textbf{r})$ must be normalized to $1$,
\begin{eqnarray}
 \label{eq:ReCon1}G^{(1)}[W]=\int d^{2N} r \,W(\mathbf{r})=1.
\end{eqnarray}
The second set of constraints demand that $W(\mathbf{r})$ and $W_G(\mathbf{r})$ have the same displacement vector with components $\{d_i\},\ i=1,\ldots, 2 N $,
\begin{eqnarray}
\label{eq:ReCon2}G_{i}^{(2)}[W] = \int d^{2N}  r \,W(\textbf{r}) r_i=d_i.
\end{eqnarray}
The last set of constraints are functionals ensuring that $W(\mathbf{r})$ and $W_G(\mathbf{r})$ have the same covariance matrices with elements $V_{ij},\ i,j=1,\ldots, 2 N $,
\begin{align}
\label{eq:ReCon3}	G_{ij}^{(3)}[W] = \int d^{2N}  r\, W(\textbf{r}) r_i r_j=V_{ij}+d_i d_j.
\end{align}
We impose a fourth constraint,
\begin{eqnarray}
\label{eq:ReCon4}    \lim_{\mathbf{r}\rightarrow \mathbf{r}_0}W(\textbf{r})=W(\textbf{r}_0), \forall\ \mathbf{r}_0 \in \mathbb{R}^{N},
\end{eqnarray}
i.e., $W(\mathbf{r})$ must be a continuous function, as a necessary criterion for $W(\textbf{r})$ to correspond to a physical state~\cite{Katerina2009}. This last constraint is not imposed in the form of a Lagrange multiplier, instead, it will be used to filter out non-physical critical points.

Equations~\eqref{eq:ReCon2} and~\eqref{eq:ReCon3} ensure that $W_G(\mathbf{r})$ is indeed the Gaussian associate of $W(\textbf{r})$. Setting up said relation between $W(\textbf{r})$ and $W_G(\textbf{r})$ through the constraints allows one to treat $W_G(\textbf{r})$ in Eq.~\eqref{eq:FuncRe1} as being any Gaussian Wigner function, i.e., not necessarily the Gaussian associate of $W(\textbf{r})$; the constraints will impose the desired relation between $W(\textbf{r})$ and $W_G(\textbf{r})$. Therefore, when performing the functional derivative on Eq.~\eqref{eq:FuncRe1}, all terms containing only $W_G(\mathbf{r})$, will disappear. Therefore, we can rewrite Eq.~\eqref{eq:FuncRe1} as,
\begin{align}
\label{eq:FuncRe2}
    \nonumber R[W] = -\text{Re}h[W]+\lambda_1 G^{(1)}[W] \\
    +\sum_{i=1}^{N}\mu_iG_i^{(2)}[W]+\sum_{i,j=1}^{ 2 N }  \Lambda_{ij}G_{ij}^{(3)}[W].
\end{align}
In Appendix \ref{app:minReal} we perform the first functional derivative and we find all critical points, (i.e., functions $W(\mathbf{\mathbf{r}})$) that respect the constraints of Eqs.~\eqref{eq:ReCon1}--\eqref{eq:ReCon3}. These critical points are the Gaussian associate of $W(\mathbf{r})$, i.e., $W_G(\mathbf{r})$, and the \emph{partially flipped} Gaussian associate. The partially flipped Gaussian associate is any function $W_{pfG}(\mathbf{r})$ that takes on positive and negative values but in such a way that $\vert W_{pfG}(\mathbf{r})\vert=W_G(\mathbf{r})$ is the Gaussian associate of $W(\mathbf{r})$.  $W_{pfG}(\mathbf{r})$ can be ruled out by enforcing the constraint of Eq. \eqref{eq:ReCon4}, as it must be discontinuous somewhere for $\vert W_{pfG}(\mathbf{r})\vert=W_G(\mathbf{r})$ to hold. Note that the \emph{completely flipped} Gaussian associate, $W_{cfG}(\mathbf{r}) = - W_G(\mathbf{r})$, had already been ruled out prior to enforcing the constraint of Eq. \eqref{eq:ReCon4}, since it violates the normalization condition of Eq. \eqref{eq:ReCon1}.  Therefore, the only critical point that corresponds to a physical state is $W_G(\mathbf{r})$.

Subsequently, we calculate the second functional derivative of Eq. \eqref{eq:FuncRe2}, and we find that $W_G(\mathbf{r})$ corresponds to a minimum for Eq.~\eqref{eq:FuncRe2}. When the functional of Eq. \eqref{eq:FuncRe2} is evaluated for $W_G(\textbf{r})$, gives $\text{Re}\mu[W_G]=0$. The partially flipped Gaussian Wigner functions do not correspond to neither minima nor maxima.

From these observations, we conclude that the Gaussian state with first and second order moments defined by \eqref{eq:ReCon2}-\eqref{eq:ReCon3} is the only local minimum of the functional $\mathrm{Re}\mu[W]$. 
If we were to consider only Wigner-positive states, we could at this point establish that it is also a global minimum, from the concavity of $\mathrm{Re}h$ over Wigner-positive states (which follows from the concavity of $\varphi(x)=-x\ln x$ over $\mathbb{R}_+$). However, we cannot use that argument anymore when Wigner-negative states come into the picture, as the functional $\mathrm{Re}h$ then becomes neither concave nor convex (since $\varphi(x)=-x\ln\vert x\vert$ is neither concave nor convex over $\mathbb{R}$). Having a single physical minimum, does not exclude the possibility of $\text{Re}\mu[W]$ being unbounded from below. However, we expect (from numerics of Sections \ref{sec:apps} and \ref{sec:GaussianChannelNumerics}, where all $\text{Re}\mu[W]$ are non-negative) that some physicality condition, which is non-tractable to include in the analytic functional optimization, could render $\text{Re}{\mu}[W]$ bounded from below by zero.

\subsection{Imaginary part}\label{subsec:MinImag}
From Eq.~\eqref{eq:Immu}, the imaginary part is always non-negative, i.e., $\text{Im }\mu[W]\geq 0$, as it is proportional (with a positive constant) to modulo the negative volume of the Wigner function. It is worthwhile to note that this is a direct consequence of choosing a negative $k$ as the branch of the complex logarithm in Sec.~\ref{sec:definition}.

\section{Applications}
\label{sec:apps}
Bosonic quantum states that satisfy the requirements of being a qubit are of particular interest in tasks of quantum information processing. Various photonic encodings of the qubit find applications in quantum communication and computation. We analyse (1) the single rail qubit encoding~\cite{Knill2001-yg}, (2) the Schr\"{o}dinger cat state based qubits~\cite{Cochrane1999, Ralph2003}, and (3) the Gottesman-Kitaev-Preskill (GKP) qubits~\cite{Gottesman2001} which have seen renewed interest due to their error correction properties.

\subsection{Single-rail photonic qubits}
\label{sec:apps_qubit}
We study the behavior of the measure for the single-rail encoded photonic qubit defined as $\ket{\bar{0}}=\ket{0}$ and $\ket{\bar{1}}=\ket{1}$, where $\ket{0}, \ket{1}$ are Fock states. In the chosen basis, the logical-0 has a zero nGM, whereas the logical-1 has a non-zero nGM. Exploration of the complete set of states can be done in a parameterized way as follows: we start with a qubit density operator, parametrized as $\rho(r)=r\outprod{\bar{0}}+(1-r)\outprod{\bar{1}}$ where $r\in[0,1]$ and apply the qubit rotation gate given as, 
\begin{align}
    U(\theta,\varphi)=\begin{pmatrix}
    \cos \theta/2 & e^{i\varphi} \sin \theta/2 \\
    e^{-i\varphi} \sin \theta/2 & -\cos\theta/2
    \end{pmatrix},
    \label{eq:qubit_unitary}
\end{align}
where varying $\theta\in (0,\pi)$ and $\varphi \in (0,2\pi)$, in principle let us explore the complete space of quantum states spanned by these basis vectors. Since states of varying purity are allowed (i.e. by tuning $r\in[0,1]$) the state space may be thought of as a Bloch ball. We simplify the analysis by considering only $\varphi=0$, i.e.\ density operators that fill a hemispherical slice of the Bloch ball. This simplification does not harm generality since states lying at the same radial and azimuthal position of the Bloch ball (i.e. same $r,\theta$) are equivalent up to a rotation in phase space. Consequently, they evaluate to the same value of the nGM. Therefore, we write,
\begin{eqnarray}
U(\theta,\varphi) \equiv U(\theta).
\end{eqnarray}

\begin{figure*}[!]
    \includegraphics[width=\textwidth]{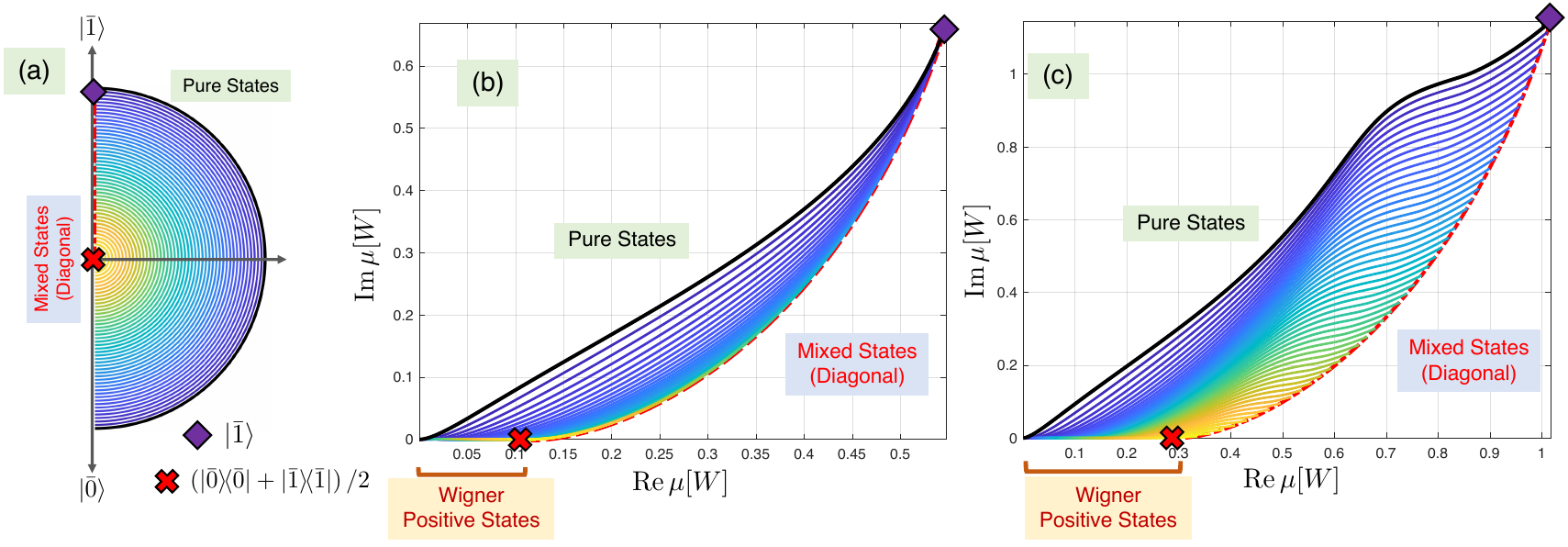}
    \caption{Non-Gaussianity measure evaluated for qubit states in the Fock basis encoding. States are initialized as $\rho_i=r\outprod{\bar{0}}+(1-r)\outprod{\bar{1}}$ and manipulated by applying the unitary $U(\theta)$ from Eq.~\eqref{eq:qubit_unitary}; the different colored tracks correspond to different values of $r$ (a) Structure of the Bloch ball; the corresponding set of measure values when (b) $\ket{\bar{0}}=\ket{0}; \ket{\bar{1}}=\ket{1}$ (c) $\ket{\bar{0}}=\ket{0}; \ket{\bar{1}}=\ket{2}$. The black line corresponds to pure qubit states, the red dashed line corresponds to states of the form $\rho=r\outprod{\bar{0}}+(1-r)\outprod{\bar{1}}; r\in[0,0.5]$. The logical-1 state (purple diamond) and completely mixed state (red cross) are marked as well.}
    \label{FockQStates}
\end{figure*}
 

Figure~\ref{FockQStates}(a) depicts tracks in the Bloch ball's hemisphere corresponding to varying $r$ (marked by different colors). The corresponding values of the measures are shown in  Fig.~\ref{FockQStates}(b) where each $\rho(r); r\in[0.5,1]$ with $U(\theta); \theta\in(0,\pi]$ applied has the corresponding colored line. Hence starting from $\rho(r)$, the final state under the present choice of evolution is $\rho(1-r)$ since $U(\pi/2)$ is the Pauli $\hat{X}$ gate. Since $r\in[0.5,1]$ for some values of $\theta$ yields passive states (i.e. Fock-diagonal states whose eigenvalues are in non-increasing order), their measures are real-valued (marked as Wigner positive states). The measures of the $\ket{\bar{1}}$ state (purple diamond) and the maximally mixed state (red cross) are marked for clarity. 

We note an upper envelope to the set of the evaluated measure values which corresponds to the pure states (black line). The set also has a right envelope which corresponds to states of the form $\rho=r\outprod{\bar{0}}+(1-r)\outprod{\bar{1}}; r\in[0,0.5]$ (red dashed line). As highlighted earlier, the complementary cases, i.e. $r\in[0.5,1]$, yields passive states.
Similar trends are observed for the case where we assume $\ket{\bar{1}}=\ket{2}$ in Fig.~\ref{FockQStates}. 

It is important to note that states with non-diagonal density matrices might have a real $\mu[W]$ since the set of Wigner positive states has elements which are non-passive. The ability to distinguish between Gaussian, positive non-Gaussian (passive), and non-Gaussian states with negative volume simultaneously is unique to our measure. One would need to take into account multiple other non-Gaussianity measures concurrently to achieve the same discern-ability as our measure.

\subsection{Application to engineering Schrödinger cat states}
\label{sec:apps_cat}
In this section we showcase our measure's potential as a computational tool to aid photonic state engineering, by comparing the values of our measure for a target state and a state produced in the engineering protocol of Ref.~\cite{Andrew2021}. In this engineering protocol, a collection of single-mode squeezed vacuum states are sent through a general passive Gaussian unitary made of beam-splitters and phase shifters. At the output of the unitary all but one of the modes is detected by a photon number resolving detector. The squeezing parameters of the input states as well as the transmissivities and phases of the general passive unitary are numerically optimized with respect to the final single mode output state's fidelity to a target state. For this example, the target states are Schrödinger cat states~\cite{Cochrane1999, Ralph2003},
\begin{align}
\label{eq:Cats}
    \ket{\Phi_t^{\pm}} = \frac{\ket{\alpha} \pm \ket{-\alpha}}{\sqrt{2(1\pm e^{-2|\alpha|^{2}})}}.
\end{align}
Here $\ket{\pm \alpha} \approx e^{-\frac{1}{2}|\alpha|^{2}}\sum_{n = 0}^{n_{c}} \frac{(\pm \alpha)^{n}}{\sqrt{n!}} \ket{n}$ is a photon number-truncated coherent state with a cutoff Fock number, $n_c = 40$ is chosen such that for a coherent amplitude up to $|\alpha| = 3$, the state is supported with high precision. 
Without loss of generality, we work with $\alpha \in \mathbb{R}$ as the phase $\arg\alpha$ for $\alpha\in\mathbb{C}$ is imparted by Gaussian unitary operation which ensures the invariance of the nGM. 

In Fig.~\ref{CatStates} we plot $\mu[W]$ for even, $\ket{\Phi_t^+}$, and odd, $\ket{\Phi_t^-}$, cat states. It is straightforward to show from Eq.~\eqref{eq:Cats} that at $\alpha = 0$, $\ket{\Phi_t^{+}} = \ket{0}$ and $\ket{\Phi_t^{-}} = \ket{1}$. This fact is reflected in Fig.~\ref{CatStates} by the even (odd) cat state's zero (non-zero) value for the real and imaginary parts of $\mu[W]$ at $\alpha = 0$. As $\alpha$ increases the even and odd cat states' measures become progressively indistinguishable. 

In Fig.~\ref{EngineeredCatStates}a we plot $\mu[W]$ for the ``circuit" states generated by the scheme in Ref.~\cite{Andrew2021} with the highest fidelity to their corresponding target states $\ket{\Phi_t^{+}}$ (see Ref.~\cite{Andrew2021}, Fig. 4 therein) along with $\mu[W]$ for these target states. In Fig. \ref{EngineeredCatStates}b we plot the difference in the target and circuit states' values for $\text{Re}\mu[W]$ and $\text{Im}\mu[W]$ as functions of the coherent amplitude of the target state  $\alpha$. Listed next to each plot point is the fidelity between the target and circuit state. Keep in mind we chose our target and circuit states to be pure, so that any non-Gaussianity reflected in the measure must arise from quantum effects and not a non-Gaussian mixture of Gaussian states.

 \begin{figure}[ht]
    \includegraphics[width=\columnwidth]{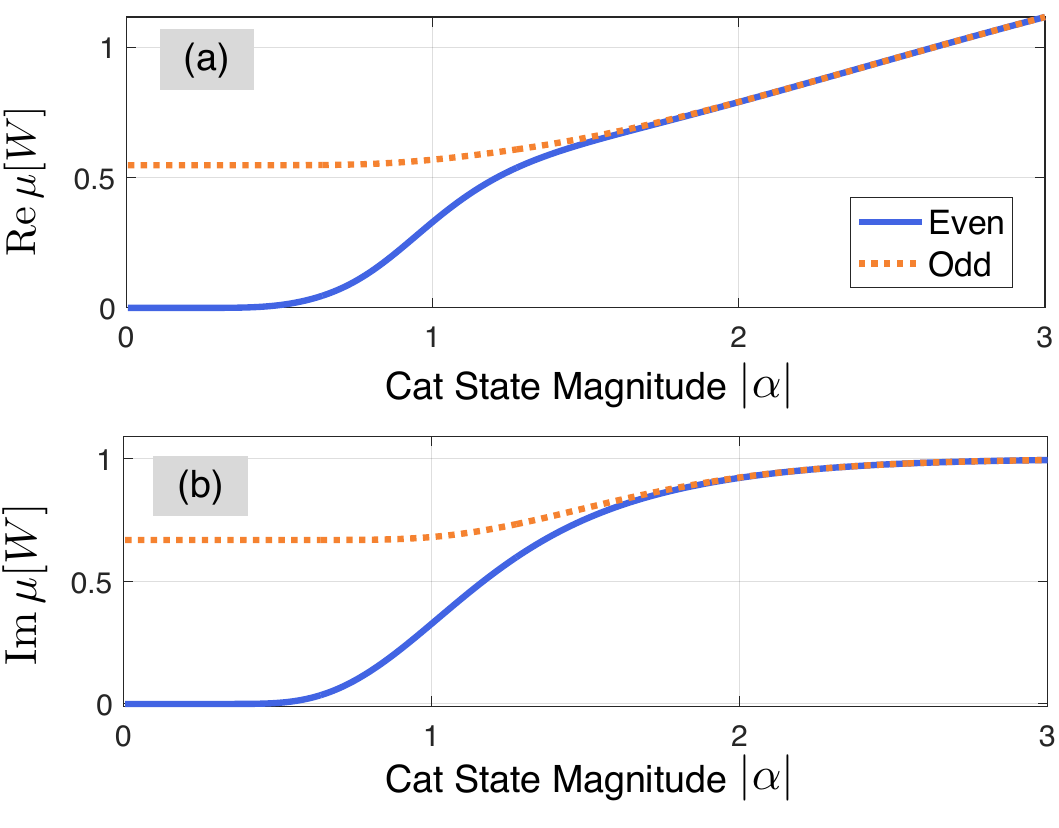}
    \caption{The (a) real and (b) imaginary parts of $\mu[W]$ evaluated for even (solid blue) and odd (dashed orange) Schr\"{o}dinger cat states as defined in Eq. \eqref{eq:Cats}. }
    \label{CatStates}
\end{figure}

From the two plots in Fig.~\ref{EngineeredCatStates} we see that the higher the fidelity between the target and the circuit states, the closer in $\mu[W]$ they are. Now, it must be pointed out that high fidelity between target and circuit states does not always mean similar $\mu[W]$ values and our results in Fig.~\ref{EngineeredCatStates} should in no way be taken as evidence to a wider connection between fidelity and our measure. Fig.~\ref{EngineeredCatStates} does provide evidence that fidelity performed well as a figure of merit for engineering these cat states. That is, not only do the target and circuit states have high fidelity, they also have a similar shape in phase space. The scenario where fidelity does not perform well as a figure of merit is where our measure should prove most valuable, something we briefly address in Section \ref{sec:concl} and leave futher exploration of to future work.


 \begin{figure}[!ht]
    \centering
    \includegraphics[width=0.9\columnwidth]{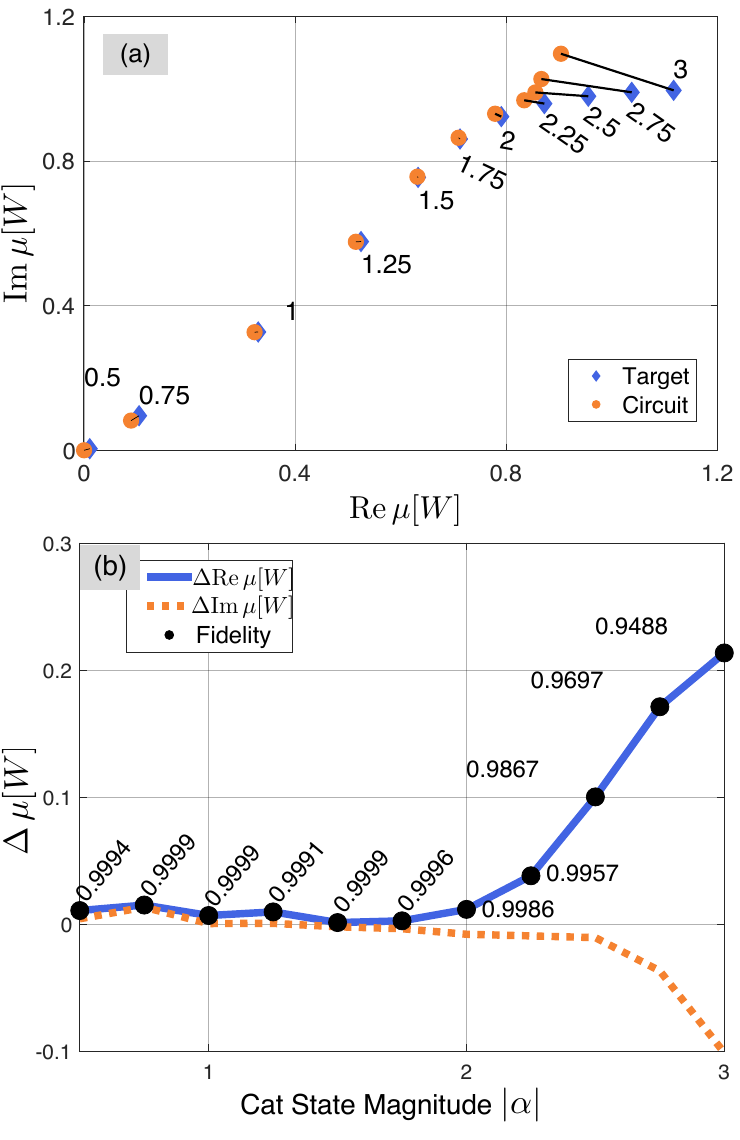}
    \caption{(a) $\mu[W]$  evaluated for target even Schr\"{o}dinger cat states  $\ket{\Phi_t^{+}}$ (blue diamonds) and ``circuit" states (orange circles) generated by the scheme in Ref.~\cite{Andrew2021} with the highest fidelity to their corresponding target state. Each circuit state point is connected to its corresponding target state point by a black line and listed next to each circuit-target state point pair is the coherent amplitude of the target state $\alpha$. (b) Difference in the target and circuit states' values for the real(imaginary) part of our measure $\Delta\text{Re}\mu[W]$ ($\Delta\text{Im}\mu[W]$) vs. the coherent amplitude of the target state $\alpha$. Listed next to each plot point is the fidelity between the target and circuit state for that particular $\alpha$. We note that while the fidelities are listed next to the points on the line for $\Delta\text{Re}\mu[W]$ (solid blue) they are meant to apply to the complementary point on the line for $\Delta\text{Im}\mu[W]$ (dashed orange) as well.}
    \label{EngineeredCatStates}
\end{figure}

\subsection{Application to engineering photonic Gottesman-Kitaev-Preskill states}\label{sec:apps_GKP}
The Gottesman-Kitaev-Preskill (GKP)~\cite{Gottesman2001} states are well known for their utility in quantum communication, error correction and computing. Here we focus on the finite squeezed square-grid GKP states~\cite{Terhal2016,Seshadreesan2022}, which do not have singularities and infinite extent on the phase space. We choose the following definition for the GKP logical states,
\begin{subequations}
\begin{align}
    &\ket{\bar{0}}_{\mathrm{GKP}}=\sum_{t=-\infty}^{\infty} e^{-\frac{\pi}{2}{\Delta}^2 (2t)^2} \hat{D}(2t\sqrt{\pi})\hat{S}(\xi)\ket{\mathrm{0}}\\
    &\ket{\bar{1}}_{\mathrm{GKP}}=\sum_{t=-\infty}^{\infty} e^{-\frac{\pi}{2} {\Delta}^2 (2t+1)^2} \hat{D}((2t+1)\sqrt{\pi})\hat{S}(\xi)\ket{\mathrm{0}},
\end{align}
    \label{eq:GKP_def}
\end{subequations}
where $\hat{D}$ is the displacement operator, $\hat{S}$ is the single-mode squeezing operator, and $\xi=-\ln \Delta$ is the squeezing parameter. The corresponding Wigner functions may be derived by calculating the symmetrically ordered characteristic function $\chi_W(\eta,\eta^*)=\mathrm{tr} [\hat{D}(\eta)\outprod{\bar{L}}],\ \bar{L}=0,1,$ and evaluating its Fourier transform.  The detailed derivation of the state Wigner function is given in Appendix~\ref{app:GKP_wigner}. The nGM of these states are shown in Fig.~\ref{fig:GKPstates}.

\begin{figure}[!]
    \centering
    \includegraphics[width=\columnwidth]{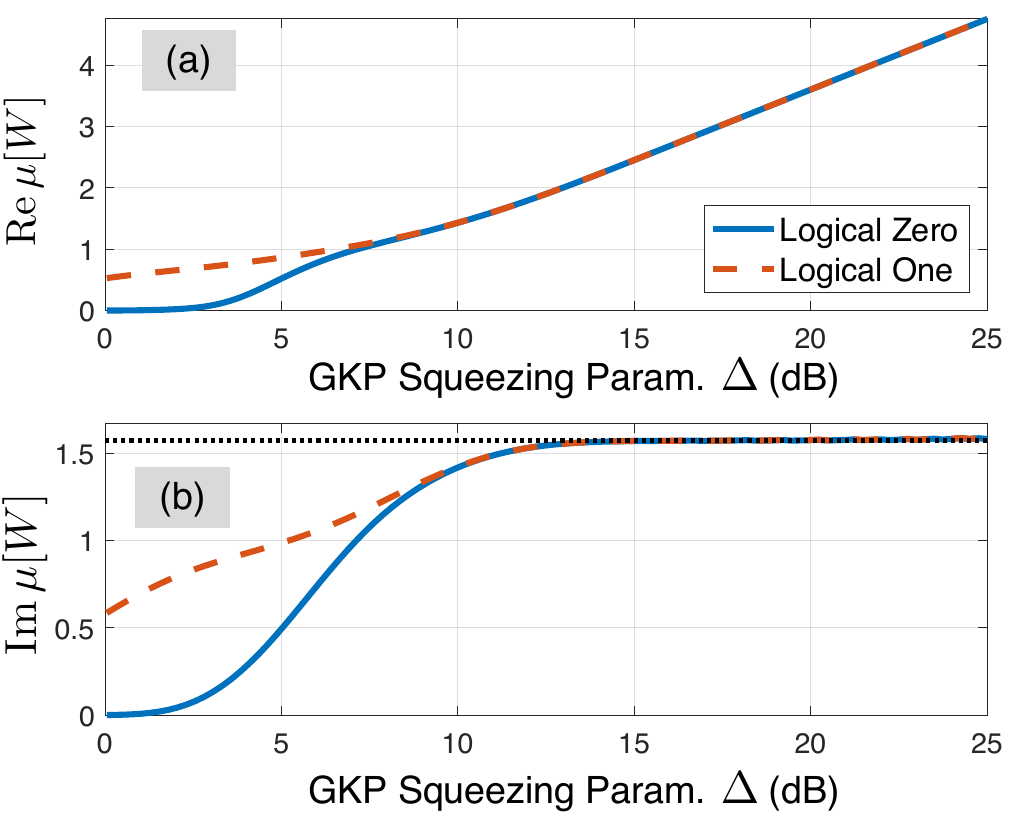}
    \caption{Non-Gaussianity measure evaluated for the finitely squeezed GKP states (Eq.~\eqref{eq:GKP_def}) for the logical zero $\ket{\bar{0}}_{\mathrm{GKP}}$ (blue solid) and logical one $\ket{\bar{1}}_{\mathrm{GKP}}$ (red dashed) states w.r.t.\ the squeezing parameter $\Delta$ (dB). }
    \label{fig:GKPstates}
\end{figure}

We note the following trends for the evaluated nGM---(1) the imaginary part of $\mu[W]$ is upper bounded for both $\ket{\bar{0}}_{\mathrm{GKP}}$ and $\ket{\bar{1}}_{\mathrm{GKP}}$, showing numerical equivalence for $\Delta\gtrsim 10$ dB and,  (2) the real part of the measure monotonically increases with $\Delta$ for both logical states showing numerical equivalence for  $\Delta\gtrsim15 $ dB.
The upper bound to $\mathrm{Im}\mu[W]$ arises due to the Wigner functions having equivalent negative regions. The upper bound ($=\pi/2$) is a result of the maximum negative volume of the GKP state which is shown to be $1/2$~\cite{Yamasaki2020}. For $\Delta\approx 0$ (low squeezing limit) we note that $\ket{\bar{0}}_{\mathrm{GKP}}\equiv \ket{0}$ (vacuum) and $\ket{\bar{1}}_{\mathrm{GKP}}\equiv\ket{\beta}+\ket{-\beta}; \beta= \sqrt{\pi}$ (Schrödinger
cat state). Hence the measure is non-zero for $\ket{\bar{1}}_{\mathrm{GKP}}$. For the  $\Delta>15$ dB regime (high squeezing), the states are expected to be displacement invariant (i.e. $\ket{\bar{1}}_{\mathrm{GKP}}=\hat{D}(\sqrt{\pi})\ket{\bar{0}}_{\mathrm{GKP}}$), thus evaluating to the same value of the measure. While it is clear that the measure for the two logical states should deviate as $\Delta \rightarrow 0 $, since the states no longer differ by a displacement, minimal insight about the physical implications of real part of the measure is currently available so the exact functional dependence of $\textrm{Re}\mu[W]$ on $\Delta$ in the low squeezing limit is not intuitive.
In Ref. \cite{Tzitrin2020} the authors introduce a photonic state engineering system that consists of sending $N$ displaced squeezed states through an $N$-mode interferometer and preforming PNR detection on $N-1$ of the modes. They denote the PNR measurement pattern by $\overline{n}$. For a given $\overline{n}$ they optimize the displacement and squeezing of the input states and the parameters of the $N$-mode interferometer to target ``approximate" GKP states. The authors define ``approximate" GKP states to be those of Eqs.~\eqref{eq:GKP_def} but with Fock support truncated at some number $n_{\text{max}}$. In Table II of Ref.~\cite{Tzitrin2020} the authors provide the fidelity between their target, approximate logical zero GKP states with $\Delta = 10$dB, and circuit states for varying $n_{\text{max}}$, $N$, and $\overline{n}$. The Fock coefficients of these target and circuit states can be found in Ref. \cite{Sabapathy2020}. We restate the results of the $n_{\text{max}} = 12$ section of aforementioned table in Table \ref{EngineeredGKPTable} with the difference in the target and circuit states’ values for $\text{Re}\mu[W]$ and $\text{Im}\mu[W]$ listed as well. Once again, our measure acts as a sort of sanity check, confirming that not only do the engineered states have high fidelity to the target but also have the desired GKP phase space characteristics.

\begin{table}[h]
\begin{tabular}{P{0.5cm} P{1.75cm} P{1.5cm} P{1.75cm} P{1cm}}
    & & $n_{\max}=12$ & &\\
    \toprule
    $N$ & 1-Fidelity & $\Delta $Re$\mu[W]$ & $\Delta $Im$\mu[W]$ & $\bar{n}$\\
    \midrule \midrule
    2 & 0.35 & 1.33 & -1.43 & (12) \\
    \midrule
     3 & $3\times10^{-3}$ & $-0.001$ & 0.032 & (5,7) \\
     \midrule
     4 & $4\times10^{-8}$ & $8.62\times10^{-4}$ & 0.003 & (3,3,6) \\
      & $2\times10^{-5}$ & $8.59\times10^{-4}$ & 0.001 & (2,4,6) \\
    \midrule
     5 & $7\times10^{-9}$ & $7.69\times10^{-4}$ & 0.003 & (1,1,3,7) \\
      & $7\times10^{-8}$ & $7.87\times10^{-4}$ & 0.003 & (1,2,3,6) \\
      \bottomrule
\end{tabular}
\caption{Results for optimizing the state engineering system in Ref. \cite{Tzitrin2020} to target approximate logical zero GKP states with $\Delta = 10$ dB and $n_{\text{max}}= 12$ which is equivalent to Eq. (\ref{eq:GKP_def}a) with Fock support truncated at $12$. Listed for each case of number of input modes, $N$, and PNR measurement pattern, $\overline{n}$, is the fidelity between the target and circuit state, presented as $1-$Fidelity, and the difference in the target and circuit states’ values for $\text{Re}\,\mu[W]$ and $\text{Im}\,\mu[W]$.}
\label{EngineeredGKPTable}
\end{table}

\section{Conclusions}\label{sec:concl}
In this work we ventured to give an application-flavored interpretation of the recently introduced quantum Wigner entropy~\cite{Zach2021}. To this end, we introduced the WRE and defined $\mu[W]$ such that under the condition of Eq.~\eqref{eq:condition_gaussian_channel} behaves as a proper non-Gaussianity measure. Our measure can serve as a computational tool to aid in assessing the quality of non-Gaussian states, generated by partial non-Gaussian measurements on Gaussian states. Indeed, the nGM introduced in this work holds information on the negative volume of the Wigner function in a naturally emerging way. At the same time, $\text{Re}\mu[W]$ provides an extra layer of information on the non-Gaussian character of the underlying state.  One would need the combination of negativity and a relative entropy measure, such as the quantum relative entropy \cite{Genoni2008}, to simultaneously distinguish between Gaussian states, non-Gaussian states with zero negativity, and non-Gaussian states with non-zero negativity. The novelty of our measure is that it is able to distinguish between these three types of states on its own. \color{black}

As pointed out in Section \ref{sec:apps_cat}, high fidelity between a circuit and target state does not always mean the circuit state will possess the desired phase space characteristics of its target. It is in these circumstances where we think our measure will flourish. Specifically, we anticipate that our work will spark interest in non-Gaussian state engineering in the following way: Optimizing the Gaussian state under partial photon detection and the (partial) photon number pattern~\cite{Su2019}, to generate a state whose $\mu[W]$ is close to $\mu[W_{\text{GKP}}]$, where $W_{\text{GKP}}$ is the Wigner function of the physical (i.e., finite squeezing) GKP state. Examination and comparison of the error-correcting properties of the measure-optimized states (in contrast to fidelity-optimized states) is an open question for forthcoming work.

Moreover, for future theoretical endeavors, complete understanding of the properties and physical meaning of $\text{Re}\mu[W]$ is of potential interest. It is worthwhile to examine the potential use of $\mu[W]$ within the context of a non-Gaussianity resource theory (i.e. from the point of view of monotones). Lastly, we believe that the conditions of Eqs.~\eqref{eq:condition_gaussian_channel} and~\eqref{eq:SuffAndNec} are interesting in their own right; non-negative Wigner functions are known to satisfy them, however it is interesting that a plethora of partialy-negative Wigner functions (e.g. cat states and GKP states under thermal loss) also seem to numerically satisfy them, as the strong evidence presented in this work suggests. Therefore, a complete understanding of which states under the action of which Gaussian channel, respect said conditions is another possible future research path.

\begin{acknowledgements}
 The authors thank Saikat Guha (Univ. of Arizona), Eneet Kaur (Univ. of Arizona). S.C. thanks Hongji Wei (Univ. of Arizona) for useful discussions on numerical methods. C.N.G., A.J.P., and S.C. acknowledge funding support from the National Science Foundation, FET, Award No. 2122337. P.D. acknowledges the Nicolaas Bloembergen Graduate Student Scholarship in Optical Sciences and the National Science Foundation (NSF) Engineering Research Center for Quantum Networks (CQN), awarded under cooperative agreement number 1941583, for supporting this research. Z.V.H. acknowledges support from the Belgian  American  Educational  Foundation. Numerical simulations were based upon High Performance Computing (HPC) resources supported by the University of Arizona TRIF, UITS, and Research, Innovation, and Impact (RII) and maintained by the UArizona Research Technologies department.
\end{acknowledgements}

\section*{Data Availability Statement}
The data that support the findings of this study are openly available in Complex-Non-Gaussianity-Measure-Data, at https://github.com/ajpizzimenti/Complex-Non-Gaussianity-Measure-Data.

\newpage
\onecolumngrid
\appendix

\section{Evolution of the non-Gaussianity measure in Gaussian channels}
\label{apd:evolution_measure_gaussian_channels}

In this Appendix, we derive a sufficient condition for the non-Gaussianity measure to decrease under the action of any Gaussian channel.
The condition is then shown to be always satisfied for states with non-negative Wigner functions.
Note that it is proven in \cite{Albarelli2018,Zhuang2018} that the negative volume can only decrease under the action of a Gaussian channel, so that we only condiser the real part of the non-Gaussianity measure here, \textit{i.e.} $\mathrm{Re}\mu[W]$.

\subsection{Multimode bosonic Gaussian channels}

The mathematical description of multimode Gaussian channels is provided in Ref. \cite{Caruso2008-dt}. 
In Wigner space, multimode Gaussian channels act as the succession of a rescaling operation and a convolution with a Gaussian distribution (for completeness we should also include a displacement, but we omit it here since it has no effect on our measure).
We define the multimode rescaling operator $\mathcal{L}_{\mathbf{S}}$ and the convolution as follows:
\begin{align}
    \left(A\ast B\right)(\mathbf{r})
    &=
    \iint A(\mathbf{r}-\mathbf{r}')B(\mathbf{r}') d^{2N}r'.
    \\
    \mathcal{L}_\mathbf{S}\left[W\right](\mathbf{r})
    &=
    \frac{1}{\vert\det\mathbf{S}\vert}W\left(\mathbf{S}^{-1}\mathbf{r}\right)
\end{align}
Note that when the rescaling matrix is proportional to the identity (so that $\mathbf{S}=s\mathbf{I}$), we will simply note the rescaling operator as $\mathcal{L}_s$ (in place of $\mathcal{L}_{s\mathbf{I}}$).
Using these two operations, we can then express any Gaussian channel $\mathcal{M}$ as:
\begin{align}
\mathcal{M}
    \left[W\right]
    &=\mathcal{L}_{\mathbf{S}}\left[W\ast \mathcal{U}\right],
\end{align}
where $\mathbf{S}$ is the rescaling matrix and $\mathcal{U}$ is some multimode Gaussian distribution.
Note that in order to ensure the physicality of the channel $\mathcal{M}$, $\mathbf{S}$ and $\mathcal{U}$ should satisfy some conditions (see \cite{Caruso2008-dt}).
As an example, a noisy lossy channel (with transmittance $\eta$ and thermal environment of $\bar{n}$ photons) corresponds to choosing $\mathbf{S}=\sqrt{\eta}\mathbf{I}_2$ and ${\mathcal{U}=\mathcal{L}_{\sqrt{(1-\eta)/\eta}}[W_{\bar{n}}]}$ where $W_{\bar{n}}=\exp(-(x^2+p^2)/(2\bar{n}+1))/(\pi(2\bar{n}+1))$ is the Wigner function of a thermal state.
Similarly, a noisy amplifier corresponds to $\mathbf{S}=\sqrt{g}\mathbf{I}_2$ and ${\mathcal{U}=\mathcal{L}_{\sqrt{(g-1)/g}}[W_{\bar{n}}]}$.

It is a simple derivation to show that the rescaling operator adds a constant to the differential entropy of $W$:
\begin{align}
    h\left[\mathcal{L}_{\mathbf{S}}\left[W\right]\right]  
    =
    h[W]
    +
    \ln(\vert\det\mathbf{S}\vert).
\end{align}
Since the non-Gaussianity is a difference of two entropies (and since the rescaling acts identically on $W$ and its Gaussian associate), it directly follows that the non-Gaussianity measure is invariant under rescaling:
\begin{align}
    \mu[\mathcal{L}_\mathbf{S}\left[W\right]]
    =
    h[\mathcal{L}_\mathbf{S}\left[W\right]]
    -
    h[\mathcal{L}_\mathbf{S}\left[W_G\right]]
    =
    h[W]
    -
    h[W_G]
    =
    \mu[W].
\end{align}
As a consequence, the evolution of the non-Gaussianity measure under a Gaussian channel is solely affected by the convolution with the Gaussian distribution.

\subsection{Infinitesimal Gaussian convolution}

We are now going to show that it is sufficient to prove that the measure is decreasing under an infinitesimal Gaussian convolution.
Indeed, let us observe that any Gaussian distribution $\mathcal{U}$ (with covariance matrix $\mathbf{U}$) can be decomposed into an arbitrary number of convolutions (assuming $1/\varepsilon\in\mathbb{N}$):
\begin{align}
    \mathcal{U}
    =
    \underbrace{
   \mathcal{L}_{\sqrt{\varepsilon}}
   \left[\mathcal{U}\right]
   \ast
   \mathcal{L}_{\sqrt{\varepsilon}}
   \left[\mathcal{U}\right]
   \ast
   \cdots
   \mathcal{L}_{\sqrt{\varepsilon}}
   \left[\mathcal{U}\right]
   }_{1/\varepsilon\text{ times}}.
\end{align}
This is a consequence of the fact that covariance matrices add up under convolution, and that the rescaling operator $\mathcal{L}_{\sqrt{\varepsilon}}$ applies a Gaussian distribution with covariance matrix $\mathbf{U}$   onto another Gaussian distribution with covariance matrix $\varepsilon \mathbf{U}$.
As it appears, when $\varepsilon$ tends towards zero, the Gaussian distribution $\mathcal{U}$ can be decomposed into (an infinite number of) convolutions of arbitrarily narrow Gaussian distributions (with covariance matrices $\varepsilon  \mathbf{U}$).
Proving that the non-Gaussianity measure increases in the limiting case $\varepsilon\rightarrow 0$ would then imply that it increases for any multiple of $\varepsilon  \mathbf{U}$, and in particular for $\mathcal{U}$ (having covariance matrix $\mathbf{U}$).
The limiting case $\varepsilon\rightarrow 0$ corresponds to taking the derivative of the measure, $\textit{i.e.}$:
\begin{align}
    \frac{\mathrm{d}}{\mathrm{d}\varepsilon}
    \mathrm{Re}\mu\left[W\ast\mathcal{L}_{\sqrt{\varepsilon}}\left[\mathcal{U}\right]\right]
    \Big\vert_{\varepsilon=0}
    \leq 0.
    \label{eq:measure_inequality_gaussian_channels}
\end{align}
We are going to focus on the latter condition in the following of this Appendix.

In the following, we use the fact the the convolution of two Gaussians is a Gaussian and that covariance matrices add up under convolution.
We use the notation $h_G[W]$ to denote the entropy of the Gaussian associate of $W$.
\begin{align}
    \mu [W\ast\mathcal{L}_{\sqrt{\varepsilon}}\left[\mathcal{U}\right]]
    &=
    h_G[W\ast\mathcal{L}_{\sqrt{\varepsilon}}\left[\mathcal{U}\right]]
    -
    h[W\ast\mathcal{L}_{\sqrt{\varepsilon}}\left[\mathcal{U}\right]]
    \\
    &=
    h[W_G\ast\mathcal{L}_{\sqrt{\varepsilon}}\left[\mathcal{U}\right]]
    -
    h[W\ast\mathcal{L}_{\sqrt{\varepsilon}}\left[\mathcal{U}\right]]
    \\
    &=
    N\ln(2\pi e)+\frac{1}{2}\ln\Big(\det(\mathbf{V}+\varepsilon \mathbf{U})\Big)
    -
    h[W\ast\mathcal{L}_{\sqrt{\varepsilon}}[\mathcal{U}]]
\end{align}

We now take the derivative with respect to $\varepsilon$ on both sides.
\begin{align}
    \frac{\mathrm{d}}{\mathrm{d}\varepsilon}
    \mu[W\ast\mathcal{L}_{\sqrt{\varepsilon}}\left[\mathcal{U}\right]]
    &=
    \frac{1}{2}
    \frac{\mathrm{d}}{\mathrm{d}\varepsilon}
    \ln\Big(\det(\mathbf{V}+\varepsilon\mathbf{U})\Big)
    -
    \frac{\mathrm{d}}{\mathrm{d}\varepsilon}
    h[W\ast\mathcal{L}_{\sqrt{\varepsilon}}[\mathcal{U}]]
    \nonumber
    \\
    &=
    \frac{1}{2}
    \left(\det(\mathbf{V}+\varepsilon\mathbf{U})\right)^{-1}
    \frac{\mathrm{d}}{\mathrm{d}\varepsilon}
    \det(\mathbf{V}+\varepsilon\mathbf{U})
    -
    \frac{\mathrm{d}}{\mathrm{d}\varepsilon}
    h[W\ast\mathcal{L}_{\sqrt{\varepsilon}}[\mathcal{U}]]
    \nonumber
    \\
    &=
    \frac{1}{2}
    \left(\det(\mathbf{V}+\varepsilon\mathbf{U})\right)^{-1}
    \det(\mathbf{V}+\varepsilon\mathbf{U})
    \Tr\left[
    (\mathbf{V}+\varepsilon\mathbf{U})^{-1}
    \frac{\mathrm{d}}{\mathrm{d}\varepsilon}
    (\mathbf{V}+\varepsilon\mathbf{U})
    \right]
    -
    \frac{\mathrm{d}}{\mathrm{d}\varepsilon}
    h[W\ast\mathcal{L}_{\sqrt{\varepsilon}}[\mathcal{U}]]
    \nonumber
    \\
    &=
    \frac{1}{2}
    \Tr\left[
    \left(\mathbf{V}+\varepsilon\mathbf{U}\right)^{-1}\mathbf{U}
    \right]
    -
    \frac{\mathrm{d}}{\mathrm{d}\varepsilon}
    h[W\ast\mathcal{L}_{\sqrt{\varepsilon}}[\mathcal{U}]]
    \nonumber
\end{align}
where we have used Jacobi's formula for the derivative of the determinant of a matrix.
Finally, we need to evaluate that expression in the limit case $\varepsilon=0$, which yields:
\begin{align}
    \frac{\mathrm{d}}{\mathrm{d}\varepsilon}
    \mu[W\ast\mathcal{L}_{\sqrt{\varepsilon}}\left[\mathcal{U}\right]]
    \Big\vert_{\varepsilon=0}
    =
    \frac{1}{2}
    \Tr\left[
    \mathbf{V}^{-1}\mathbf{U}
    \right]
    -
    \frac{\mathrm{d}}{\mathrm{d}\varepsilon}
    h[W\ast\mathcal{L}_{\sqrt{\varepsilon}}[\mathcal{U}]] \Big\vert_{\varepsilon=0}.
    \label{eq:derivative_measure_epsilon}
\end{align}

\subsection{De Bruijn's identity for Wigner functions}

The second term in the RHS of Eq. \eqref{eq:derivative_measure_epsilon} corresponds to the increase of entropy of $W$ when it undergoes an infinitesimal Gaussian convolution.
If $W$ was a genuine probability distribution, that quantity would be related to the Fisher information through the so-called De Bruijn identity \cite{Stam1959-tq}.
In this section, we show that De Bruijn's identity  can be extended to Wigner functions, even when they take negative values.
We provide here a concise proof, which closely follows the the usual derivation of De Bruijn's identity \cite{Cover1999-rp,Johnson2004-dj}.

Let us first observe that solutions of the heat equation $\partial W/\partial \varepsilon=(1/2)\Delta W$ (where $\Delta=\sum_i(\partial/\partial r_i)^2$) have the form:
\begin{align}
    W_\varepsilon=W_0\ast\mathcal{L}_{\sqrt{\varepsilon}}\left[N\right],
\end{align}
where $N$ is a standard Gaussian disitribution of covariance matrix $\mathbf{I}$.
Then, we can compute the derivative of $\mathrm{Re}  h(W_\varepsilon)$ with respect to $\varepsilon$ as follows:
\begin{align}
    \frac{\mathrm{d}}{\mathrm{d}\varepsilon}
    \mathrm{Re}h[W_\varepsilon]
    &=
    -\int \frac{\mathrm{d}}{\mathrm{d}\varepsilon}\Big(W_\varepsilon(\mathbf{r})\ln\vert W_\varepsilon(\mathbf{r})\vert\Big)
     d^{2N}r
    \nonumber
    \\
    &=
    -\int
    \Big(
    \ln\vert W_\varepsilon(\mathbf{r})\vert
    +1
    \Big)
    \frac{\mathrm{d}}{\mathrm{d}\varepsilon}W_\varepsilon(\mathbf{r})
     d^{2N}r
    \nonumber
    \\
    &=
    -\frac{1}{2}\sum_{i}
    \int
    \Big(
    \ln\vert W_\varepsilon(\mathbf{r})\vert
    +1
    \Big)
    \frac{\partial^2}{\partial r_i^2}W_\varepsilon(\mathbf{r})
     d^{2N}r
    \nonumber
    \\
    &=
    -\frac{1}{2}\sum_{i}
    \left[
    \Big(
    \ln\vert W_\varepsilon(\mathbf{r})\vert
    +1
    \Big)
    \left(
    \frac{\partial}{\partial r_i}W_\varepsilon(\mathbf{r})
    \right)
    \right]_{r_i=-\infty}^{+\infty}
    +
    \frac{1}{2}
    \sum_{i}
    \int
    \left(
    \frac{\partial}{\partial r_i}
    \ln\vert W_\varepsilon(\mathbf{r})\vert
    \right)
    \left(
    \frac{\partial}{\partial r_i}W_\varepsilon(\mathbf{r})
    \right)
     d^{2N}r
    \nonumber
    \\
    &=
    \frac{1}{2}
    \sum_i
    \int
    \left(
    \frac{1}{W_\varepsilon(\mathbf{r})}
    \frac{\partial}{\partial r_i}W(\mathbf{r})
    \right)^2
    W_\varepsilon(\mathbf{r}) d^{2N}r
    \label{eq:standard_debruijn}
\end{align}
Here, we have used integration by parts.
Then, from a similar argument as presented in \cite{Cover1999-rp}, we concluded that the term in square brackets is zero.
The next step of the proof is then to generalize to Gaussian distribution which are non-standard (going from $N$ to any Gaussian $\mathcal{U}$).
To that purpose, we introduce the Fisher information matrix $\mathbf{J}$, which is defined from $W$ as follows:
\begin{align}
    J_{ij}\equiv \left(\mathbf{J}\right)_{ij}
    &=
    \int
    \left(
    \frac{1}{W(\mathbf{r})}
    \frac{\partial}{\partial r_i}
    W(\mathbf{r})
    \right)
    \left(
    \frac{1}{W(\mathbf{r})}
    \frac{\partial}{\partial r_j}
    W(\mathbf{r})
    \right)
    W(\mathbf{r}) d^{2N}r.
    \label{eq:FIwrtW}
\end{align}

Then, following a similar argument as the one presented in \cite{Johnson2004-dj}, we extend the proof to the case of non-standard Gaussian distributions, which gives us:
\begin{align}
    \frac{\mathrm{d}}{\mathrm{d}\varepsilon}
    \mathrm{Re}h\left[W\ast\mathcal{L}_{\sqrt{\varepsilon}}\left[\mathcal{U}\right]\right]
    \Big\vert_{\varepsilon=0}
    =
    \frac{1}{2}
    \Tr\left[
    \mathbf{U}\mathbf{J}\right].
    \label{eq:multivariate_debruijn}
\end{align}
where $\mathbf{J}$ is the Fisher information matrix of $W$ and $\mathbf{U}$ is the covariance matrix of the Gaussian $\mathcal{U}$ (note that if $\mathbf{U}=\mathbf{I}$ we find Eq. \eqref{eq:standard_debruijn}).
Eq. \eqref{eq:multivariate_debruijn} is precisely the multivariate De Bruijn identity \cite{Costa1984-az, Johnson2004-rt}, and we have shown here that it holds even for distributions that take negative values (under the condition that we extend usual entropy to the real part of the complex-valued entropy).
Using this relation, we can finally rewrite Eq. \eqref{eq:derivative_measure_epsilon} as:
\begin{align}
    \frac{\mathrm{d}}{\mathrm{d}\varepsilon}
    \mathrm{Re}\mu\left[
    W\ast\mathcal{L}_{\sqrt{\varepsilon}}\left[\mathcal{U}\right]    
    \right]\Big\vert_{\varepsilon=0}
    &=
    \frac{1}{2}
    \Tr\left[
    \mathbf{U}(\mathbf{V}^{-1}-\mathbf{J})
    \right].
    \label{eq:derivative_measure_with_debruijn}
\end{align}
If we want the latter expression to be negative for any Gaussian channel (and thus for any covariance matrix $\mathbf{U}$), we need the matrix $(\mathbf{V}^{-1}-\mathbf{J})$ to be negative semi-definite.

\subsection{A proof for non-negative Wigner functions with the Cram\'er-Rao bound}

In the particular case of Wigner-positive states (having non-negative Wigner functions), we can prove that Eq. \eqref{eq:measure_inequality_gaussian_channels} always holds.
First, we will use the fact that the Fisher information matrix is always positive semi-definite (for genuine probability distributions).
Then, we will use the multivariate Cram\'er-Rao bound \cite{Bobrovsky1987-vt}, which states the following:
\begin{align}
    \mathbf{V}\geq\mathbf{J}^{-1},
\end{align}
and should be understood as $\left(\mathbf{V}-\mathbf{J}^{-1}\right)$ being a positive semi-definite matrix.
Then, since $ \mathbf{V} \geq \mathbf{J}^{-1} > 0 $ we have that $ \mathbf{J} \geq \mathbf{V}^{-1} > 0 $ and thus $(\mathbf{V}^{-1}-\mathbf{J}) \leq 0$, negative semi-definite.    \cite[Corollary 7.7.4 (a)]{horn2012matrix}. From \eqref{eq:derivative_measure_with_debruijn} we find:
\begin{align}
    \frac{\mathrm{d}}{\mathrm{d}\varepsilon}
    \mathrm{Re}\mu\left[W\ast\mathcal{L}_{\sqrt{\varepsilon}}\left[\mathcal{U}\right]\right]
    \Big\vert_{\varepsilon=0}
    &=
    \frac{1}{2}
    \Tr\left[
    \mathbf{U}\left(\mathbf{V}^{-1}-\mathbf{J}\right)
    \right]
    \\
    &=
    \frac{1}{2}
    \Tr\left[
    \mathbf{U}^{\frac{1}{2}}\left(\mathbf{V}^{-1}-\mathbf{J}\right)\mathbf{U}^{\frac{1}{2}}
    \right]
    \\& \leq 0,
\end{align}
where we have used the fact that the matrix $\mathbf{U}^{\frac{1}{2}}\left(\mathbf{V}^{-1}-\mathbf{J}\right)\mathbf{U}^{\frac{1}{2}}$ is negative semi-definite, and that the trace of a negative semi-definite matrix is non-positive.

In conclusion, we have proven that if a state has a non-negative Wigner function, then its non-Gaussianity measure $\mathrm{Re}\mu[W]$ can only decrease when the state evolves in a Gaussian channel.
However, if the state has a Wigner function that takes negative values, we couldn't use the Cram\'er-Rao bound to conclude our proof.

\section{Numerical Methods}
\label{app:NumericalMethods}

The Wigner function of a Fock state $|n\rangle$ is,
\begin{align}
W_n(q,p)=\frac{e^{-(q^2+p^2)}}{\pi}(-1)^nL_n(2q^2+2p^2).
\end{align}
In polar coordinates $q=r \sin\theta,\ p=r \cos \theta$, i.e.,  $r^2=q^2+p^2$, we can write,
\begin{align}
W_n(r,\theta)\equiv W_n(r)=\frac{e^{-r^2}}{\pi}(-1)^nL_n(2r^2),
\end{align}
where $L_n(.)$ is the Laguerre polynomial of the $n-$th order.
From the definition of Fisher information $\mathbf{J}$ we get,
\begin{subequations}
\begin{align}
\text{Tr}(\mathbf{J})&=\int^{2\pi}_0\int^{+\infty}_0 W_n(r)(\Vec{\nabla} \ln|W_n(r)|)^2 r dr d\theta\\
&=2\pi\int^{+\infty}_0 W_n(r)\left(\left(\frac{\partial}{\partial r}\ln|W_n(r)|\right)^2+\left(\frac{1}{r}\frac{\partial}{\partial\theta} \ln|W_n(r)|\right)^2\right) r dr\\
&=2\pi\int^{+\infty}_0 \frac{r}{W_n(r)}\left(\frac{\partial W_n(r)}{\partial r}\right)^2 dr\\
&=\int^{+\infty}_0 \frac{8(-1)^n e^{-r^2}r^3(L_n(2r^2)+2L_{n-1}^1(2r^2))^2}{L_n(2r^2)} dr\\
&=\int^{+\infty}_0 \frac{2(-1)^n e^{-r^2}r^2(L_n(2r^2)+2L_{n-1}^1(2r^2))^2}{L_n(2r^2)} d(2r^2)
\end{align}
\end{subequations}
By letting $x=2r^2$ we can write,
\begin{align}
\label{eqapp:J}\text{Tr}(\mathbf{J})=\int^{+\infty}_0 \frac{(-1)^n e^{-\frac{x}{2}}x(L_n(x)+2L_{n-1}^1(x))^2}{L_n(x)} dx,
\end{align}
where $L_n^\alpha(.)$ are the associated Laguerre polynomials with $L_n^0(.)=L_n(.)$. The main challenge of Eq. \eqref{eqapp:J} is that the denominator can be equal to zero. Therefore, we calculate the principal value of said integral. To this end, we rewrite Eq. \eqref{eqapp:J} as,
\begin{subequations}
\begin{align}
\text{Tr}(\mathbf{J})&=\int^{+\infty}_0 \frac{(-1)^n e^{-\frac{x}{2}}x((L_n(x))^2+4L_n(x)L_{n-1}^1(x)+(2L_{n-1}^1(x))^2)}{L_n(x)} dx\\
\label{eq:NGFI_term_a} &=\int^{+\infty}_0 (-1)^n e^{-\frac{x}{2}}x(L_n(x)+4L_{n-1}^1(x)) dx\\
&\quad \label{eq:NGFI_term_b} +\int^{+\infty}_0 \frac{(-1)^n e^{-\frac{x}{2}}x\cdot 4(L_{n-1}^1(x))^2)}{L_n(x)} dx.
\end{align}
\end{subequations}
Now, we consider the two terms separately,
\begin{subequations}
\begin{align}
\eqref{eq:NGFI_term_a}&=\int^{+\infty}_0 (-1)^n e^{-\frac{x}{2}}x L_n(x) dx + \int^{+\infty}_0 (-1)^n e^{-\frac{x}{2}}x\cdot4L_{n-1}^1(x) dx\\
&=(-1)^n\int^{+\infty}_0  e^{-\frac{x}{2}}x L_n(x) dx + 4(-1)^n\int^{+\infty}_0  e^{-\frac{x}{2}}x L_{n-1}^1(x) dx.
\end{align}
\end{subequations}
We find that,
\begin{align}
\int^{+\infty}_0 e^{-st}t^\beta L^\alpha_n(t) dt&=\frac{\Gamma(\beta+1)\Gamma(\alpha+n+1)}{n!\Gamma(\alpha+1)}s^{-\beta-1}{_2F_1(-n, \beta+1, \alpha+1,\frac{1}{s})}\equiv f(n,\beta,\alpha,s)
\end{align}
where $_2F_1(.)$ and $\Gamma(.)$ are the confluent hypergeometric function and the gamma function respectively. Therefore,
\begin{align}
\label{eq:NGFI_term_a_sol} \eqref{eq:NGFI_term_a} =(-1)^n f(n,1,0,\frac{1}{2})+4(-1)^n f(n-1,1,1,\frac{1}{2})=-8n+4.
\end{align}
Next, we consider the term \eqref{eq:NGFI_term_b}. The Laguerre polynomials can be written as,
\begin{align}
\label{eq:NGFI_Ln} L_n(x)&=\sum^n_{k=0}
\begin{pmatrix}
n\\k
\end{pmatrix}
\frac{(-1)^k}{k!}x^k,\\
\label{eq:NGFI_Ln-1} L^1_{n-1}(x)&=\sum^{n-1}_{k=0}
\begin{pmatrix}
n\\k+1
\end{pmatrix}
\frac{(-1)^k}{k!}x^k
\end{align}
For the term $x(L^1_{n-1}(x))^2$ in \eqref{eq:NGFI_term_b} we write,
\begin{align}
\label{eq:NGFI_pq} x(L^1_{n-1}(x))^2=p(x)L_n(x)+q(x),
\end{align}
where,
\begin{align}
\label{eq:NGFI_p} p(x)&=\sum^{n-1}_{k=0}p_k x^k,\\
\label{eq:NGFI_q} q(x)&=\sum^{n-1}_{k=0}q_k x^k
\end{align}
To calculate $p_k$ and $q_k$, we can use the commands \texttt{PolynomialQuotient}, \texttt{PolynomialMod} and \texttt{CoefficientList} in Mathematica. Then, we can separate \eqref{eq:NGFI_term_b} into two parts,
\begin{subequations}
\begin{align}
\eqref{eq:NGFI_term_b}&=4(-1)^n\int^{+\infty}_0 \frac{e^{-\frac{x}{2}}(p(x)L_n(x)+q(x))}{L_n(x)} dx\\
\label{eq:NGFI_term_1} &=4(-1)^n\int^{+\infty}_0 p(x) e^{-\frac{x}{2}} dx\\
&\quad \label{eq:NGFI_term_2} +4(-1)^n\int^{+\infty}_0 \frac{q(x)}{L_n(x)} e^{-\frac{x}{2}} dx
\end{align}
\end{subequations}
Now, \eqref{eq:NGFI_term_1} is rendered manageable. By virtue of \eqref{eq:NGFI_p}, we get
\begin{subequations}
\begin{align}
\eqref{eq:NGFI_term_1}=&4(-1)^n\int^{+\infty}_0 \sum^{n-1}_{k=0}p_k x^k e^{-\frac{x}{2}} dx\\
=&4(-1)^n\sum^{n-1}_{k=0}p_k\int^{+\infty}_0  x^k e^{-\frac{x}{2}} dx\\
\stackrel{\text{let $x=2t$}}{=\joinrel=}&4(-1)^n\sum^{n-1}_{k=0}p_k\int^{+\infty}_0  (2t)^k e^{-t} d(2t)\\
=&4(-1)^n\sum^{n-1}_{k=0}p_k 2^{k+1}\int^{+\infty}_0  t^k e^{-t} dt\\
=&4(-1)^n\sum^{n-1}_{k=0}p_k 2^{k+1}\Gamma(k+1)\\
\label{eq:NGFI_term_1_sol} =&4(-1)^n\sum^{n-1}_{k=0}p_k 2^{k+1}k!
\end{align}
\end{subequations}
Now, we consider the term \eqref{eq:NGFI_term_2}. To compute said term, let us rewrite $\frac{1}{L_n(x)}$ using \eqref{eq:NGFI_Ln},
\begin{subequations}
\begin{align}
\frac{1}{L_n(x)}
=\frac{1}{\sum^n_{k=0}
\begin{pmatrix}
n\\k
\end{pmatrix}
\frac{(-1)^k}{k!}x^k}
&=\frac{1}{a_0 x^0+...+a_n x^n}\\
&=\frac{1}{a_n(x-x_1)...(x-x_n)}\\
&=\frac{1}{a_n}\prod^n_{i=1}\frac{1}{x-x_i}
\end{align}
\end{subequations}
where $x_i$ is the root of $L_n(x)=0$, and $a_0,\ldots,a_n$ are the coefficients of increasing powers of $x$ of the polynomial $L_n(x)$. Therefore,
\begin{subequations}
\begin{align}
a_n&=
\begin{pmatrix}
n\\n
\end{pmatrix}
\frac{(-1)^n}{n!}=\frac{(-1)^n}{n!},\\
\label{eq:NGFI_a} a& \equiv \frac{1}{a_n}=(-1)^n n!,\\
\frac{1}{L_n(x)}&=a\prod^n_{i=1}\frac{1}{x-x_i},
\end{align}
\end{subequations}
Next, we recast the product into a summation,
\begin{subequations}
\begin{align}
\prod^n_{i=1}\frac{1}{x-x_i}&=\sum^n_{i=1}\frac{b_i}{x-x_i},\\
\lim_{x\rightarrow x_k} (x-x_k)\prod^n_{i=1}\frac{1}{x-x_i}&=\lim_{x\rightarrow x_k} (x-x_k)\sum^n_{i=i}\frac{b_i}{x-x_i},\\
\lim_{x\rightarrow x_k}\prod^n_{\substack{i=1\\i\neq k}}\frac{1}{x-x_i}&=\lim_{x\rightarrow x_k} \frac{x-x_k}{x-x_k}b_k,\\
\label{eq:NGFI_b} b_k&=\prod^n_{\substack{i=1\\i\neq k}}\frac{1}{x_k-x_i}.
\end{align}
\end{subequations}
Therefore,
\begin{align}
\label{eq:NGFI_1/Ln} \frac{1}{L_n(x)}&=a\sum^n_{i=1}\frac{b_i}{x-x_i}.
\end{align}
Now, using Eq. \eqref{eq:NGFI_1/Ln}, term \eqref{eq:NGFI_term_2} can be written as,
\begin{subequations}
\begin{align}
\eqref{eq:NGFI_term_2}&=4(-1)^n\int^{+\infty}_0 a\sum^n_{i=1}\frac{b_i q(x)}{x-x_i}e^{-\frac{x}{2}} dx\\
&=4(-1)^n a\sum^n_{i=1}b_i\int^{+\infty}_0 \frac{ q(x)}{x-x_i}e^{-\frac{x}{2}} dx\\
\label{eq:NGFI_term_I} &= 4(-1)^n a\sum^n_{i=1}b_i\int^{+\infty}_0 \frac{ q(x_i)}{x-x_i}e^{-\frac{x}{2}} dx\\
&\quad \label{eq:NGFI_term_II} +4(-1)^n a\sum^n_{i=1}b_i\int^{+\infty}_0 \frac{ q(x)-q(x_i)}{x-x_i}e^{-\frac{x}{2}} dx,
\end{align}
\end{subequations}
where we made the poles to become apparent in \eqref{eq:NGFI_term_II}. We simplify \eqref{eq:NGFI_term_I} as,
\begin{subequations}
\begin{align}
\eqref{eq:NGFI_term_I}&=4(-1)^n a\sum^n_{i=1}b_i q(x_i)\int^{+\infty}_0 \frac{e^{-\frac{x}{2}}}{x-x_i} dx\\
\label{eq:NGFI_term_I_sol} &= 4(-1)^n a\sum^n_{i=1}b_i q(x_i)E_i(x_i)
\end{align}
\end{subequations}
Where $E_i(x_i)$ is,
\begin{align}
\label{eq:NGFI_Ei} E_i(x_i) = \int^{+\infty}_0 \frac{e^{-\frac{x}{2}}}{x-x_i} dx
\end{align}
and represented by \texttt{ExpIntegralEi} in Mathematica.

Using Eq. \eqref{eq:NGFI_q}, term \eqref{eq:NGFI_term_II} gives,
\begin{subequations}
\begin{align}
\eqref{eq:NGFI_term_II}&=4(-1)^n a\sum^n_{i=1}b_i\int^{+\infty}_0 \frac{\sum^{n-1}_{k=0}q_k (x^k-x_i^k)}{x-x_i}e^{-\frac{x}{2}} dx\\
&=4(-1)^n a\sum^n_{i=1}b_i\sum^{n-1}_{k=0}q_k\int^{+\infty}_0 \frac{ x^k-x_i^k}{x-x_i}e^{-\frac{x}{2}} dx\\
\label{eq:NGFI_term_II_1} &=4(-1)^n a\sum^n_{i=1}b_i\sum^{n-1}_{k=1}q_k\int^{+\infty}_0 \frac{ x^k-x_i^k}{x-x_i}e^{-\frac{x}{2}} dx\\
&\quad \label{eq:NGFI_term_II_2} +4(-1)^n a\sum^n_{i=1}b_i q_0\int^{+\infty}_0 \frac{0}{x-x_i}e^{-\frac{x}{2}} dx.
\end{align}
\end{subequations}
Term \eqref{eq:NGFI_term_II_2} is equal to $0$. Therefore, we focus on term \eqref{eq:NGFI_term_II_1}. To this end, we first observe that, 
\begin{align}
a^k-b^k=(a-b)(a^{k-1}+a^{k-2}b+...+ab^{k-2}+b^{k-1}).
\end{align}
Therefore,
\begin{align}
\label{eq:NGFI_a-b} x^k-x_i^k=(x-x_i)\sum^{k-1}_{m=0}x^{k-1-m}x_i^m.
\end{align}
Using Eq. \eqref{eq:NGFI_a-b}, the term \eqref{eq:NGFI_term_II_1} gives,
\begin{subequations}
\begin{align}
\eqref{eq:NGFI_term_II_1}=&4(-1)^n a\sum^n_{i=1}b_i\sum^{n-1}_{k=1}q_k\int^{+\infty}_0 \frac{ (x-x_i)\sum^{k-1}_{m=0}x^{k-1-m}x_i^m}{x-x_i}e^{-\frac{x}{2}} dx\\
=&4(-1)^n a\sum^n_{i=1}b_i\sum^{n-1}_{k=1}q_k\sum^{k-1}_{m=0}x_i^m \int^{+\infty}_0 x^{k-1-m}e^{-\frac{x}{2}} dx\\
\stackrel{\text{let $x=2t$}}{=\joinrel=}&4(-1)^n a\sum^n_{i=1}b_i\sum^{n-1}_{k=1}q_k\sum^{k-1}_{m=0}x_i^m \int^{+\infty}_0 (2t)^{k-m-1}e^{-t} d(2t)\\
=&4(-1)^n a\sum^n_{i=1}b_i\sum^{n-1}_{k=1}q_k\sum^{k-1}_{m=0}x_i^m 2^{k-m} \int^{+\infty}_0 t^{k-m-1}e^{-t} dt\\
=&4(-1)^n a\sum^n_{i=1}b_i\sum^{n-1}_{k=1}q_k\sum^{k-1}_{m=0}x_i^m 2^{k-m} \Gamma(k-m)\\
\label{eq:NGFI_term_II_1_sol} =&4(-1)^n a\sum^n_{i=1}b_i\sum^{n-1}_{k=1}q_k\sum^{k-1}_{m=0}x_i^m 2^{k-m} (k-m-1)!
\end{align}
\end{subequations}

Finally, we write $\text{Tr}(\mathbf{J})=$\eqref{eq:NGFI_term_a_sol}+\eqref{eq:NGFI_term_1_sol}+\eqref{eq:NGFI_term_I_sol}+\eqref{eq:NGFI_term_II_1_sol} based on Eqs. \eqref{eq:NGFI_pq}, \eqref{eq:NGFI_p}, \eqref{eq:NGFI_q}, \eqref{eq:NGFI_a}, \eqref{eq:NGFI_b},  \eqref{eq:NGFI_Ei}, \eqref{eq:NGFI_Ln}, and \eqref{eq:NGFI_Ln-1}. That is,
\begin{subequations}
\begin{align}
\text{Tr}(\mathbf{J})&=-8n+4\\
&+4(-1)^n\sum^{n-1}_{k=0}p_k 2^{k+1}k!\\
&+4(-1)^n a\sum^n_{i=1}b_i q(x_i)E_i(x_i)\\
&+4(-1)^n a\sum^n_{i=1}b_i\sum^{n-1}_{k=1}q_k\sum^{k-1}_{m=0}x_i^m 2^{k-m} (k-m-1)!
\end{align}
\end{subequations}
where,
\begin{subequations}
\begin{align}
p(x)L_n(x)+q(x)&=x(L^1_{n-1}(x))^2\\
p(x)&=\sum^{n-1}_{k=0}p_k x^k\\
q(x)&=\sum^{n-1}_{k=0}q_k x^k\\
a&=(-1)^n n!\\
b_k&=\prod^n_{\substack{i=1\\i\neq k}}\frac{1}{x_k-x_i}\\
E_i(x_i) &= \int^{+\infty}_0 \frac{e^{-\frac{x}{2}}}{x-x_i} dx\\
L_n(x)&=\sum^n_{k=0}
\begin{pmatrix}
n\\k
\end{pmatrix}
\frac{(-1)^k}{k!}x^k\\
L^1_{n-1}(x)&=\sum^{n-1}_{k=0}
\begin{pmatrix}
n\\k+1
\end{pmatrix}
\frac{(-1)^k}{k!}x^k\\
\text{$x_i$ } &\text{is the $i-$th root of $L_n(x)=0$}
\end{align}
\end{subequations}
\section{Critical points of \texorpdfstring{$\mu[W]$}{mu[W]}}
\label{app:minReal}

Let us revise some basic elements of functional derivatives (for a concise overview see also \cite[Appendix C therein]{BookForAppC}). The functional we will consider is of the following integral form,
\begin{eqnarray}
\label{eq:functionalGeneral} F[W(\mathbf{r})]=\int d^{2N} r g(W(\mathbf{r}))+\sum_{i=1}^{ 2 K} \lambda_i G^{(i)}[W(\mathbf{r})],
\end{eqnarray}
where $ g(W(\mathbf{r}))$ is not a functional, but rather a form that involves $W(\mathbf{r})$, e.g., $g(W(\mathbf{r}))=-W(\mathbf{r}) \ln |W(\mathbf{r})|$ and $G^{(i)}[W(\mathbf{r})]$ are functionals representing $K$ constraints, i.e., to find the critical points of $F[W(\mathbf{r})]$ over normalized Wigner functions the constraint $G^{(1)}[W(\mathbf{r})]=\int d^{2N} r  W(\mathbf{r})=1$ should be included.

Let us assume that the functionals $G_i[W(\mathbf{r})]$ are in an integral form and therefore Eq. \eqref{eq:functionalGeneral} can be written as,
\begin{eqnarray}
\label{eq:functionalGeneral2} F[W(\mathbf{r})]=\int d^{2N} r  f(W(\mathbf{r})).
\end{eqnarray}
The functional derivative is defined through the (usual) derivative,
\begin{eqnarray}
\label{eq:FuncDervt1}\delta F[W(\mathbf{r})] =\int d^{2N} r  \frac{d f(W(\mathbf{r})+\varepsilon h(\mathbf{r}))}{d\varepsilon}\Bigg|_{\varepsilon\rightarrow 0},
\end{eqnarray}
where $h(\mathbf{r})$ is \emph{any} function. After performing the derivative with respect to $\varepsilon$ and taking the limit $\varepsilon\rightarrow 0$, if  Eq. \eqref{eq:FuncDervt1} assumes the form,
\begin{eqnarray}
\label{eq:FuncDervt2}\delta F[W(\mathbf{r})] =\int d^{2N} r  \frac{d f(W(\mathbf{r})+\varepsilon h(\mathbf{r}))}{d\varepsilon}\Bigg|_{\varepsilon\rightarrow 0}= \int d^{2N} r  h(\mathbf{r}) \frac{\delta F[W(\mathbf{r})]}{\delta W(\mathbf{r})},
\end{eqnarray}
i.e., if the function $h(\mathbf{r})$ factors out, then the term multiplying $h(\mathbf{r})$, denoted as $\frac{\delta F[W(\mathbf{r})]}{\delta W(\mathbf{r})}$, is called the functional derivative of the functional $F[W(\mathbf{r})]$ with respect to the function $W(\mathbf{r})$. Otherwise, the functional derivative cannot be defined.

Setting Eq. \eqref{eq:FuncDervt2} equal to zero, and solving for $W(\mathbf{r})$ will give the Wigner functions that serve as critical points of the functional $F[W(\mathbf{r})]$. We have,
\begin{eqnarray}
\nonumber\delta F[W(\mathbf{r})] &=&\int d^{2N} r   h(\mathbf{r}) \frac{\delta F[W(\mathbf{r})]}{\delta W(\mathbf{r})}=0\Rightarrow\\
\label{eq:FuncDervt3} \frac{\delta F[W(\mathbf{r})]}{\delta W(\mathbf{r})}&=&0,
\end{eqnarray}
where the last step is valid since $h(\mathbf{r})$ is any function.
The critical points will depend on $\lambda_i$. To identify the $\lambda_i$, one takes the solutions of Eq. \eqref{eq:FuncDervt3}, and plugs them into the constraints which are typically given by taking the (usual) derivatives of $F[W(\mathbf{r})]$ with respect to each $\lambda_i$.

Similarly, we can define the second functional derivative by starting with,
\begin{eqnarray}
\label{eq:secFuncDervt1}\delta^2 F[W(\mathbf{r})] =\int d^{2N} r  \frac{d^2 f(W(\mathbf{r})+\varepsilon h(\mathbf{r}))}{d\varepsilon^2}\Bigg|_{\varepsilon\rightarrow 0}.
\end{eqnarray}
If Eq. \eqref{eq:secFuncDervt1} can be written as,
\begin{eqnarray}
\label{eq:secFuncDervt2}\delta^2 F[W(\mathbf{r})] =\int d^{2N} r  \frac{d^2 f(W(\mathbf{r})+\varepsilon h(\mathbf{r}))}{d\varepsilon^2}\Bigg|_{\varepsilon\rightarrow 0}= \int d^{2N} r  h^2(\mathbf{r}) \frac{\delta^2 F[W(\mathbf{r})]}{\delta W(\mathbf{r})^2},
\end{eqnarray}
then the second functional derivative exist and is equal to $\frac{\delta^2 F[W(\mathbf{r})]}{\delta W(\mathbf{r})^2}$.

Applying this recipe to Eq. \eqref{eq:FuncRe1} (or Eq. \eqref{eq:FuncRe2}) we find,
\begin{eqnarray}
\label{eq:FuncDrvtRe1}	\frac{\delta R[W(\mathbf{r})]}{\delta W(\mathbf{r})}  = 1+\lambda_1+\ln |W(\textbf{r})|+\sum_{i=1}^{ 2 N } \mu_i r_{i}+\sum_{i,j=1}^{ 2 N } \Lambda_{ij}r_{i} r_{j}
\end{eqnarray}

Setting Eq. \eqref{eq:FuncDrvtRe1} equal to zero and using vector-vector and vector-matrix multiplication (for compactness), we get,
\begin{eqnarray}
\label{eq:RePartWigner1}	|W^{\text{CP}}(\textbf{r})| = \frac{1}{|k|} \exp\left[-\left(\mathbf{r}^T \boldsymbol{\Lambda} \mathbf{r}+\boldsymbol{\mu}^T \mathbf{r}\right)\right],
\end{eqnarray}
where $\boldsymbol{\Lambda}$ is a matrix with elements $\Lambda_{ij}$, $\boldsymbol{\mu}$ is a vector with elements $\mu_i$, and $|k|=e^{\lambda_1+1}$.
Equation \eqref{eq:RePartWigner1} gives the Wigner function that are the critical points the functional of Eqs. \eqref{eq:FuncRe1} and \eqref{eq:FuncRe2}. The set of solutions satisfying Eq. \eqref{eq:RePartWigner1}  includes non-physical Wigner functions. To see this, we observe that the right hand side of \eqref{eq:RePartWigner1} is always positive as it is an exponential function of real numbers. To remove the absolute value we have two options: we either remove it by equating $W(\textbf{r})$ to the right hand side, or we equate $W(\textbf{r})$ to a \emph{partially flipped} version of the exponential, i.e., an $\frac{1}{|k|} \exp\left[-\left(\mathbf{r}^T \boldsymbol{\Lambda} \mathbf{r}+\boldsymbol{\mu}^T \mathbf{r}\right)\right]$ for some phase space regions (or even individual points) and $-\frac{1}{|k|} \exp\left[-\left(\mathbf{r}^T \boldsymbol{\Lambda} \mathbf{r}+\boldsymbol{\mu}^T \mathbf{r}\right)\right]$ for the complementary regions (or points).  We demand that any physical Wigner function must be continuous, therefore the only extremal Wigner function which is also physical, is the Gaussian solution,
\begin{eqnarray}
\label{eq:RePartWigner2}	W^{\text{CP}}(\textbf{r}) = \frac{1}{|k|} \exp\left[-\left(\mathbf{r}^T \boldsymbol{\Lambda} \mathbf{r}+\boldsymbol{\mu}^T \mathbf{r}\right)\right].
\end{eqnarray}
Using the constraints of Eqs. \eqref{eq:ReCon1}, \eqref{eq:ReCon2}, and \eqref{eq:ReCon3}, we get,
\begin{eqnarray}
\boldsymbol{\Lambda}&=&\frac{1}{2} \mathbf{V}^{-1},\\
\boldsymbol{\mu}&=&-\mathbf{V}^{-1} \mathbf{d},\\
|k|&=& (2\pi)^N \sqrt{\det \mathbf{V}} \exp\left[\frac{1}{2}\mathbf{d}^T \mathbf{V}^{-1}\mathbf{d}\right].
\end{eqnarray}
Therefore the only physical critical point is,
\begin{eqnarray}
\label{eq:RePartWigner3} W^{\text{CP}}(\textbf{r}) = \frac{1}{(2\pi)^N \sqrt{\det \mathbf{V}}} \exp\left[-\frac{1}{2}(\mathbf{r}-\mathbf{d})^T \mathbf{V}^{-1} (\mathbf{r}-\mathbf{d})\right]\equiv W_G(\mathbf{r}),
\end{eqnarray}
which is the Gaussian associate of $W(\mathbf{r})$.
For said critical point, Eq. \eqref{eq:Remu} gives $\text{Re}\mu[W_G]]=0$.

The second functional derivative of Eq. \eqref{eq:FuncRe1} or \eqref{eq:FuncRe2} is evaluated using Eq. \eqref{eq:secFuncDervt2} to be,
\begin{eqnarray}
    \frac{\delta^2 F[W(\mathbf{r})]}{\delta W(\mathbf{r})^2} =\frac{1}{W(\mathbf{r})},
\end{eqnarray}
which when evaluated at the critical point $W(\mathbf{r})=W_G(\mathbf{r})>0$, gives 
\begin{eqnarray}
\frac{\delta^2 F[W(\mathbf{r})]}{\delta W(\mathbf{r})^2}=\frac{1}{W_G(\mathbf{r})}>0,
\end{eqnarray}
therefore the Gaussian associate of $W(\mathbf{r})$ corresponds to a minimum. The partially flipped Gaussian functions do not correspond to a definite sign of the second functional derivative, therefore they are neither minima nor maxima.

\section{Wigner Function of Gottesman-Kitaev-Preskill (GKP) States}
\label{app:GKP_wigner}

The finitely squeezed GKP qubit states~\cite{Terhal2016,Seshadreesan2022} are defined in terms of a summation of displaced squeezed vacuum states as follows,
\begin{align}
	\ket{\bar{L}}=\sum_{t=-\infty}^{\infty} \exp(-\pi{\Delta}^2 (2t+\bar{L})^2/2) \hat{D}((2t+\bar{L})\sqrt{\pi})\hat{S}(\xi)\ket{\mathrm{0}}
\end{align}
where $ \bar{L}=\{0,1\} $ for the two logical qubit states. Let us use $ \rho_{\bar{L}}=\outprod{\bar{L}} $ to represent the density operators for the qubit states. We first calculate the symmetrically ordered characteristic function for $ \rho_{\bar{L}} $ defined as 
\begin{align}
    \chi_W(\eta,\eta^*) =\Tr\left(\rho_{\bar{L}}  e^{-\eta^*\hat{a}+\eta\hat{a}^\dagger } \right)=\left\langle \hat{D}(\eta)\right\rangle_{\rho_{\bar{L}}}.
\end{align} 
Before we proceed, let us make the following abbreviations
\begin{align}
	\outprod{\bar{L}}&=\sum_{t,t'=-\infty}^{\infty} e^{-\frac{\pi{\Delta}^2}{2} [(2t+\bar{L})^2+ (2t'+\bar{L})^2]} \hat{D}((2t+\bar{L})\sqrt{{\pi}/{2}})\hat{S}(\xi)\outprod{\mathrm{0}} \hat{S}^\dagger(\xi)\hat{D}^\dagger((2t'+\bar{L})\sqrt{{\pi}/{2}})\\
	&\equiv \sum_{t,t'=-\infty}^{\infty} f(t)f(t') \hat{D}(\alpha)\hat{S}(\xi)\outprod{\mathrm{0}}\hat{S}^\dagger(\xi)\hat{D}^\dagger(\beta)
\end{align}
with $f(t)=\exp(-\pi{\Delta}^2 (2t+\bar{L})^2/2), \alpha=(2t+\bar{L})\sqrt{\pi/2} \text{ and } \beta=(2t'+\bar{L})\sqrt{\pi/2} $. Simplifying $\left\langle \hat{D}(\eta)\right\rangle_{\rho_{\bar{L}}}$ yields, 
\begin{subequations}
	\begin{align}
		\left\langle \hat{D}(\eta)\right\rangle_{\rho_{\bar{L}}} =\Tr\left(\sum_{t,t'=-\infty}^\infty  f(t)f(t') \hat{D}(\alpha)\hat{S}(\xi)\outprod{\mathrm{0}}\hat{S}^\dagger(\xi)\hat{D}^\dagger(\beta) \hat{D}(\eta)\right)\\
		=\sum_{t,t'=-\infty}^\infty  f(t)f(t')\;\Tr\left( \hat{D}(\alpha)\hat{S}(\xi)\outprod{\mathrm{0}}\hat{S}^\dagger(\xi)\hat{D}^\dagger(\beta) \hat{D}(\eta)\right)\\
		=\sum_{t,t'=-\infty}^\infty  f(t)f(t')\;\Tr\left( \outprod{\mathrm{0}}\hat{S}^\dagger(\xi)\hat{D}^\dagger(\beta) \hat{D}(\eta)\hat{D}(\alpha)\hat{S}(\xi)\right)\\
		=\sum_{t,t'=-\infty}^\infty  f(t)f(t')\;\Tr\left( \outprod{\mathrm{0}}\hat{S}^\dagger(\xi)\hat{D}(\tau)\hat{S}(\xi)\right) g(\eta,\alpha,\beta)
	\end{align}
\end{subequations}
where $\tau \equiv \eta+\alpha-\beta$ and  since $ \alpha,\beta\in\mathbb{R} $, $g(\eta,\alpha,\beta)=\exp\left({(\eta-\eta^*)(\alpha+\beta)}/{2}\right)$. We may further simplify as,
\begin{subequations}
	\begin{align}
	&\sum_{t,t'=-\infty}^\infty  f(t)f(t')\;\Tr\left( \outprod{\mathrm{0}}\hat{S}^\dagger(\xi)\hat{D}(\tau)\hat{S}(\xi)\right)  g(\eta,\alpha,\beta)\\
	&=\sum_{t,t'=-\infty}^\infty  f(t)f(t')\;\Tr\left( \outprod{\mathrm{0}}\hat{S}^\dagger(\xi)\hat{S}(\xi)\hat{D}(\tau')\right)  g(\eta,\alpha,\beta)\\
	&=\sum_{t,t'=-\infty}^\infty  f(t)f(t')\;\Tr\left( \outprod{\mathrm{0}}\hat{D}(\tau' )\right)  g(\eta,\alpha,\beta)\\
	&=\sum_{t,t'=-\infty}^\infty  f(t)f(t')\exp\left({-(\eta-\eta^*)(\alpha+\beta)}/{2}\right) \exp\left(-{|\tau'|^2}/{2}\right) 
	\end{align}
\end{subequations}
where we made the substitution, $ \tau'=\tau \cosh|\xi| +\tau^*e^{i\theta} \sinh |\xi| $, as a result of reordering $\hat{S}(\xi)$ with $\hat{D}(\tau)$. Let us assume that $ \xi $ is real i.e.\ $ |\xi|=\xi; \theta=0 $ which gives 
\begin{align}
	|\tau'|^2&=|\tau|^2\left(\cosh^2\xi+\sinh^2\xi\right)+\sinh\xi \cosh\xi\,(\tau^2+\tau^{*2})
\end{align}
We see the following simplifications
\begin{subequations}
	\begin{align}
		&\tau^2= \eta^2+(\alpha-\beta)^2+2\eta(\alpha-\beta),\\
			&|\tau|^2=|\eta|^2+(\alpha-\beta)^2+(\eta+\eta^*)(\alpha-\beta).
	\end{align}
\end{subequations}
This gives us the final relation,
\begin{align}
	\chi_W(\eta)=\sum_{t,t'=-\infty}^\infty  f(t)f(t') \exp\biggl(&\frac{(\eta-\eta^*)(\alpha+\beta)}{2} 
	-\frac{|\eta|^2(\cosh^2\xi+\sinh^2\xi)}{2}-\frac{(\eta^2+\eta^{*2})(\cosh\xi\sinh\xi)}{2} \nonumber\\
&\quad	-\frac{(\eta+\eta^*)(\alpha-\beta)(\cosh \xi+\sinh \xi)^2}{2}-\frac{(\alpha-\beta)^2 (\cosh \xi+\sinh \xi)^2 }{2}\biggr).
	\label{eqn:char_simple}
\end{align}
Subsequently one may use the Fourier transform relation between $ \chi_{W}(\eta) $ and $ W(\gamma) $ as follows
\begin{align}
	W(\gamma)=\frac{1}{\pi^2}\int d^2\eta \; \chi_W(\eta) e^{-\eta\gamma^*+\eta^* \gamma},
\end{align}
This Fourier transform kernel simplifies to $ -i \sqrt{2}\eta_2 q + i \sqrt{2} \eta_1 p$, with the transform variable pairs $ \sqrt{2}q\leftrightarrow\eta_2 $ and $ -\sqrt{2}p\leftrightarrow\eta_1 $. This gives the Wigner component for each component of the sum,
\begin{align}
W_{\ket{\bar{L}}\bra{\bar{L}}}(q,p)|_{\alpha,\beta}=\frac{1}{\pi}	\exp\left(-p^2 e^{-2 \xi }+i \sqrt{2}p (\alpha -\beta )-e^{2 \xi } \biggl(\frac{1}{\sqrt{2}} (\alpha +\beta )+q\biggr)^2\right).
\end{align}
Hence the complete Wigner function may now be expressed as,
{\small
\begin{align}
	W_{\ket{\bar{L}}\bra{\bar{L}}}(q,p)=&\sum_{t,s=-\infty}^\infty \frac{1}{\pi} e^{-\frac{\pi {e}^{-2\xi}}{2} [(2s+\bar{L})^2+ (2t+\bar{L})^2]} 	\exp\left(- p^2 e^{-2 \xi }+i 2\sqrt{2}p (t-s )\sqrt{\frac{\pi}{2}}- e^{2 \xi } \biggl(\frac{1}{\sqrt{2}}(s+t+2\bar{L})\sqrt{\frac{\pi}{2}}+q\biggr)^2\right)\\
	=&\sum_{t,s=-\infty}^\infty \frac{1}{\pi} e^{-\frac{\pi{\Delta}^2}{2} [(2t+\bar{L})^2+ (2t'+\bar{L})^2]}	\exp\left(-p^2 \Delta^2+i 2\sqrt{2} p (t-s )\sqrt{\frac{\pi}{2}}- \Delta^{-2} \biggl(\frac{1}{\sqrt{2}}(s+t+2\bar{L})\sqrt{\frac{\pi}{2}}+q\biggr)^2\right).
	\label{eq:wigner_GKP}
\end{align}
}
where $ \xi=-\ln {\Delta} $. Note that the Wigner function is not properly normalized and must be done so ad-hoc. An alternative, but equivalent, definition of the GKP Wigner function can be found in Ref.~ \cite{bourassa2021fast}.\\
\\

\twocolumngrid

\bibliographystyle{abbrv}
\bibliography{bibliography_WRE}

\end{document}